\documentclass[a4paper]{article}

\usepackage{mathtools}
\usepackage{amssymb,amsthm}
\usepackage{tikz}
\usepackage{paralist}
\usepackage[basic]{complexity}
\usepackage{authblk}
\usepackage{xspace}
\usepackage{fullpage}
\usepackage{enumitem}
\usepackage{nicefrac}
\usepackage{hyperref}
\usepackage{todonotes}

\setlist{nosep}

\makeatletter
\let\epsilon\varepsilon
\makeatother

\newtheorem{remark}{Remark}
\newtheorem{theorem}{Theorem}
\newtheorem{lemma}{Lemma}
\newtheorem{proposition}{Proposition}
\newtheorem{claim}{Claim}
\newtheorem{corollary}{Corollary}

\theoremstyle{definition}

\renewclass{\AP}{APTIME}
\renewclass{\P}{PTIME}
\renewclass{\EXP}{EXPTIME}

\usetikzlibrary{positioning,fit,arrows,automata,backgrounds,calc,shapes} 
\tikzset{->,>=stealth',shorten >=1pt,shorten <=1pt,auto,node distance=1cm,
initial text={},
every fit/.style={draw,densely dotted,rectangle},
el/.style={font=\scriptsize},
inner sep=2mm,
ve/.style={rectangle, draw},
va/.style={circle, draw},
}

\newcommand{\calA}{\mathcal{A}}
\newcommand{\calB}{\mathcal{B}}
\newcommand{\calC}{\mathcal{C}}
\newcommand{\calD}{\mathcal{D}}
\newcommand{\calO}{\mathcal{O}}
\newcommand{\calF}{\mathcal{F}}

\renewcommand{\restriction}{\mathord{\upharpoonright}}

\providecommand\st{}
\renewcommand{\st}{\:{\mid}\:}
\newcommand{\pow}{\mathcal{P}}
\newcommand{\eve}{Eve\xspace}
\newcommand{\adam}{Adam\xspace}

\newcommand{\eg}{\textit{e.g.}\xspace}

\newcommand{\supfun}{\mathsf{Sup}}
\newcommand{\inffun}{\mathsf{Inf}}
\newcommand{\lsupfun}{\mathsf{LimSup}}
\newcommand{\linffun}{\mathsf{LimInf}}
\newcommand{\mpfun}{\mathsf{MP}}

\newcommand{\regret}[3]{\mathbf{reg}^{#1}_{#3}(#2)}
\newcommand{\Regret}[2]{\mathbf{Reg}_{#2}(#1)}

\newcommand{\Val}{\mathbf{Val}}
\newcommand{\aVal}{\mathbf{aVal}}
\newcommand{\cVal}{\mathbf{cVal}}
\newcommand{\PlayVal}[1]{\Val(#1)}
\newcommand{\StratVal}[4]{\Val_{#1}^{#2}(#3,#4)}
\newcommand{\out}[3]{\pi^{#1}_{#2#3}}
\newcommand{\cost}[1]{\mathbf{c}({#1})}

\newcommand{\proj}[2]{{[#1]}_{\mathbf{#2}}}
\newcommand{\projinv}[1]{\proj{#1}{1}^{-1}}
\newcommand{\VtcE}{V_\exists}
\newcommand{\VtcA}{V \setminus \VtcE}
\newcommand{\widehatVtcA}{\widehat{V} \setminus \widehat{\VtcE}}
\newcommand{\StrE}{\Sigma_\exists}
\newcommand{\StrA}{\Sigma_\forall}
\newcommand{\StrAllE}{\mathfrak{S}_\exists}
\newcommand{\StrAllA}{\mathfrak{S}_\forall}
\newcommand{\StrPosA}{\Sigma^1_\forall}
\newcommand{\StrWordA}{\mathfrak{W}_\forall}
\newcommand{\StrPosE}{\Sigma^1_\exists}
\newcommand{\range}[1]{\mathrm{Range}(#1)}
\newcommand{\outdeg}{deg^+}

\title{Reactive Synthesis Without Regret\thanks{This work was supported by the ERC inVEST (279499) project. Additionally, G. A. P\'erez had an (Aspirant) research fellowship from F.R.S-FNRS.}}

\author[1]{Paul Hunter}
\author[2]{Guillermo A. P\'erez}
\author[1]{Jean-Fran\c{c}ois Raskin}
\affil[1]{Universit\'e libre de Bruxelles, Belgium}
\affil[2]{University of Antwerp -- Flanders Make, Belgium}

\begin{document}

\maketitle

\begin{abstract}
	Two-player zero-sum games of infinite duration and their quantitative
	versions are used in verification to model the interaction between
	a controller (\eve) and its environment (\adam). The
	question usually addressed is that of the existence (and computability)
	of a strategy for \eve that can maximize her payoff against any strategy
	of \adam. In this work, we are interested in strategies of \eve that
	minimize her regret, i.e. strategies that minimize the difference
	between her actual payoff and the payoff she could have achieved if she
	had known the strategy of \adam in advance. We give algorithms to
	compute the strategies of \eve that ensure minimal regret against an
	adversary whose choice of strategy is
	\begin{inparaenum}[$(i)$]
	\item unrestricted,
	\item limited to positional strategies, or
	\item limited to word strategies,
	\end{inparaenum}
	and show that the two last cases have natural modelling applications.
	These results apply for quantitative games defined with the classical payoff
	functions $\inffun$, $\supfun$, $\linffun$, $\lsupfun$, and mean-payoff.
	We also show that our notion of regret minimization in which \adam is
	limited to word strategies generalizes the notion of good for games
	introduced by Henzinger and Piterman, and is related to the notion of
	determinization by pruning due to Aminof, Kupferman and Lampert.

\end{abstract}

\section{Introduction}
\label{sec:intro}

The model of two player games played on graphs is an adequate mathematical tool
to solve important problems in computer science, and in particular the 
reactive-system synthesis problem~\cite{pr89}. In that context, the game models the
non-terminating interaction between the system to synthesize and its
environment. Games with quantitative objectives are useful to formalize important
quantitative aspects such as mean-response time or energy consumption. They
have attracted large attention recently, see
e.g.~\cite{cdhr10,
bcdgr09}. Most of the contributions in this
context are for zero-sum games: the objective of \eve (that models the system)
is to maximize the value of the game while the objective of \adam (that models
the environment) is to minimize this value. This is a worst-case assumption:
because the cooperation of the environment cannot be assumed, we postulate that
it is {\em antagonistic}.

In this antagonistic approach, the main solution concept is that of a {\em winning
strategy}. Given a threshold value, a winning strategy for \eve ensures a
minimal value greater than the threshold against any strategy of \adam. However,
sometimes there are no winning strategies. What should the behaviour of the
system be in such cases? There are several possible
answers to this question. One is to consider {\em non-zero sum} extensions of
those games:  the environment (\adam) is not completely antagonistic, rather it has
its own specification. In such games, a strategy for \eve must be winning only
when the outcome  satisfies the objectives of \adam, see e.g.~\cite{cdfr14}.
Another option for \eve is to play a strategy which minimizes her {\em regret}.
The regret is informally defined as the difference between what a player
actually wins and what she could have won if she had known the strategy chosen
by the other player. Minimization of regret is a central concept in decision
theory~\cite{bell82}. This notion is important because it usually leads to
solutions that agree with common sense.

\begin{figure}
\begin{center}
\begin{tikzpicture}

\node[ve,initial](A){$v_1$};
\node[va,right=of A](B){$v_2$};
\node[va,below=of A](C){$v_3$};
\node[ve,right=of B](D){$v_4$};
\node[ve,right=of C](E){$v_5$};

\path
(A) edge[bend right] node[el,swap]{$\Sigma,1$} (B)
(A) edge node[el,swap]{$\Sigma,1$} (C)
(B) edge[bend right] node[el,swap]{$b,-1$} (A)
(C) edge[loop left] node[el,swap]{$b,\frac{1}{2}$} (C)
(B) edge node[el]{$a,2$} (D)
(C) edge node[el,swap]{$a,1$} (E)
(D) edge[loop right] node[el,swap]{$\Sigma,2$} (D)
(E) edge[loop right] node[el,swap]{$\Sigma,1$} (E)
;
\end{tikzpicture}
\caption{Example weighted arena $G_0$.}\label{fig:intro-ex}
\end{center}
\end{figure}

Let us illustrate the notion of regret minimization on the example of
Fig.~\ref{fig:intro-ex}. In this example, \eve owns the squares and \adam owns
the circles (we do not use the letters labelling edges for the moment). The game
is played for infinitely many rounds and the value of a play for \eve is the
long-run average of the values of edges traversed during the play (the so-called
{\em mean-payoff}). In this game, \eve is only able to secure a mean-payoff of
$\frac{1}{2}$ when \adam is fully antagonistic. Indeed, if \eve (from $v_1$)
plays to $v_2$ then \adam can force a mean-payoff value of $0$, and if she plays to
$v_3$ then the mean-payoff value is at least $\frac{1}{2}$.
Note also that if \adam is not fully antagonistic, then the mean-payoff could be
as high as $2$.  Now, assume that \eve does not try to force the highest value
in the worst-case but tries to minimize her regret. If she plays $v_1 \mapsto
v_2$ then the regret is equal to $1$. This is because \adam can play the
following strategy: if \eve plays to $v_2$ (from $v_1$) then he plays $v_2
\mapsto v_1$ (giving a mean-payoff of $0$), and if \eve plays to $v_3$ then he
plays to $v_5$ (giving a mean-payoff of $1$). If she plays  $v_1 \mapsto
v_3$ then her regret is $\frac{3}{2}$ since \adam can play
the symmetric strategy. It should thus be clear that the strategy of \eve which
always chooses $v_1 \mapsto v_2$ is indeed minimizing her regret.

In this paper, we will study three variants of {\em regret minimization}, each
corresponding to a different set of strategies we allow \adam to choose from.
The first variant is when \adam can play any possible strategy (as in the
example above), the second variant is when \adam is restricted to playing {\em
memoryless strategies}, and the third variant is when \adam is restricted to
playing {\em word strategies}. To illustrate the last two variants, let us consider
again the example of Fig.~\ref{fig:intro-ex}. Assume now that \adam is playing
memoryless strategies only. Then in this case, we claim that there is a strategy
of \eve that ensures regret $0$. The strategy is as follows: first play to
$v_2$, if \adam chooses to go back to $v_1$, then \eve should henceforth play $v_1
\mapsto v_3$.
We claim that this strategy has regret $0$. Indeed, when $v_2$ is visited,
either \adam chooses $v_2 \mapsto v_4$, and then \eve secures a mean-payoff
of $2$ (which is the maximal possible value), or \adam chooses $v_2 \mapsto v_1$
and then we know that $v_1 \mapsto v_2$ is not a good option for \eve as
cycling between $v_1$ and $v_2$ yields a payoff of only $0$.
In this case, the mean-payoff is either $1$, if \adam plays $v_3  \mapsto v_5$,
or $\frac{1}{2}$, if he plays $v_3 \mapsto v_1$. In all the cases,
the regret is $0$. Let us now turn to the restriction to word strategies for
\adam. When considering this restriction, we use the letters that label the
edges of the graph. A word strategy for \adam is a function $w : \mathbb{N}
\rightarrow \{ a,b \}$. In this setting \adam plays a sequence of letters and
this sequence is independent of the current state of the game. It is more
convenient to view the latter as a game played on a weighted automata---assumed
to be total and with at least one transition for every action from every
state---in which \adam plays letters and \eve responds by resolving
non-determinism.
When \adam plays
word strategies, the strategy that minimizes regret for \eve is to always play
$v_1 \mapsto v_2$. Indeed, for any word in which the letter $a$ appears, the
mean-payoff is equal to $2$, and the regret is $0$, and for any word in which
the letter $a$ does not appear, the mean-payoff is $0$ while it would have been
equal to $\frac{1}{2}$ when playing $v_1 \mapsto v_3$. So the regret of this
strategy is $\frac{1}{2}$ and it is the minimal regret that \eve can secure.
Note that the three different strategies give three different values in our
example. This is in contrast with the worst-case analysis of the same problem
(memoryless strategies suffice for both players).

We claim that at least the two last variants are useful for modelling purposes. For
example, the memoryless restriction is useful when designing a system that needs
to perform well in an environment which is only partially known. In practical
situations, a controller may discover the environment with which it is
interacting at run time. Such a situation can be modelled by an arena in which
choices in nodes of the environment model an entire family of environments and
each memoryless strategy models a specific environment of the family. In such
cases, if we want to design a controller that performs reasonably well against
all the possible environments, we can consider a controller that minimizes
regret: the strategy of the controller will be as close as possible to an
optimal strategy if we had known the environment beforehand.  This is, for
example, the modelling choice done in the famous Canadian traveller's
problem~\cite{py91}: a driver is attempting to reach a specific location while
ensuring the traversed distance is \emph{not too far} from the shortest feasible
path. The partial knowledge is due to some roads being closed because of
snow.
The Canadian traveller, when planning his itinerary, is in fact searching for a
strategy to minimize his regret for the shortest path measure against a
memoryless adversary who determines the roads that are closed. Similar
situations naturally arise when synthesizing controllers for \emph{robot motion
planning}~\cite{weu15}.
We now illustrate the usefulness of the variant in which \adam is restricted to
play word strategies.
Assume that we need to design a system embedded into an
environment that produces disturbances: if the sequence of disturbances produced
by the environment is independent of the behavior of the system, then it is
natural to model this sequence not as a function of the state of the system but
as a temporal sequence of events, i.e. a {\em word} on the alphabet of the
disturbances.  Clearly, if the sequences are
not the result of an antagonistic process, then minimizing the regret against
all disturbance sequences is an adequate solution concept to obtain a reasonable
system and may be preferable to a system obtained from a strategy that is
optimal under the antagonistic hypothesis.

\paragraph{Contributions.} In this paper, we provide algorithms to solve the
\emph{regret threshold problem} (strict and non-strict) in the three variants
explained above, i.e.  given a game and a threshold, does there exist a strategy
for \eve with a regret that is (strictly) less than the threshold against all
(resp. all memoryless, resp.  all word) strategies for \adam. It is worth
mentioning that, in the first two cases we consider, we actually provide
algorithms to solve the following search problem: find the controller which
ensures the minimal possible regret. Indeed, 
our algorithms are reductions to well-known games and are such that a winning
strategy for \eve in the resulting game corresponds to a regret-minimizing
strategy in the original one.
Conversely, in games played against word strategies for \adam, we only
explicitly solve the regret threshold problem. However, since the set of
possible regret values of the considered games is finite and easy to describe,
it will be obvious that one can implement a binary search to find the regret
value and a corresponding optimal regret-minimizing strategy for
\eve.

We study this problem for six common quantitative measures: $\inffun$,
$\supfun$, $\linffun$, $\lsupfun$, $\underline{{\sf MP}}$, $\overline{{\sf
MP}}$. For all measures, but $\mpfun$, the strict and non-strict threshold
problems are equivalent. We state our results for both cases for consistency.
In almost all the cases, we provide matching lower bounds showing the worst-case
optimality of our algorithms. Our results are summarized in the table of
Fig.~\ref{tab:sum}.
\begin{table}
\begin{center}
\small
\begin{tabular}{|l||c|c|c|c|}
	\hline
	Payoff type & Any strategy & Memoryless strategies & Word strategies \\
	\hline\hline
	$\supfun$, $\inffun$,	& \P-c & \PSPACE~
		(Lem~\ref{lem:pspace-memless-regret}), & \EXP-c\\
	$\lsupfun$ & (Thm~\ref{thm:anyadversary}) & $\coNP$-h
		(Lem~\ref{lem:coNP-hardness}) &
		(Thm~\ref{thm:eloquentadversary})\\
	\hline
	$\linffun$ & \P-c (Thm~\ref{thm:anyadversary}) & \PSPACE-c
		(Thm~\ref{thm:memlessadversary})
		& \EXP-c (Thm~\ref{thm:eloquentadversary})\\
	\hline
	$\underline{{\sf MP}}$, $\overline{{\sf MP}}$ & {\sf MP equiv.}
		(Thm~\ref{thm:anyadversary}) & \PSPACE-c
		(Thm~\ref{thm:memlessadversary})
		& {\sf Undecidable} (Lem~\ref{lem:undec-mpg}) \\
	\hline
\end{tabular}
\caption{Complexity of deciding the regret threshold problem.}\label{tab:sum}
\end{center}
\end{table}
For the variant in which \adam plays word strategies only, we show that we can
recover decidability of mean-payoff objectives when the memory of \eve is fixed
in advance: in this case, the problem is $\NP$-complete
(Theorems~\ref{thm:one-assumption} and~\ref{thm:np-hardness}). 

\paragraph{Related works.}
The notion of regret minimization is a central one in game theory, see
e.g.~\cite{zjbp08} and references therein.  Also, {\it iterated} regret
minimization has been recently proposed by Halpern et al. as a concept for {\em
non-zero} sum games~\cite{hp12}. There, it is applied to matrix games and not to
game graphs. In a previous contribution, we have applied the iterated regret
minimization concept to non-zero sum games played on weighted graphs for the
shortest path problem~\cite{fgr10}. Restrictions on how \adam is allowed
to play were not considered there.
As we do not consider an explicit objective for \adam, we do not consider
iteration of the regret minimization here.

The disturbance-handling embedded system example was first given in~\cite{df11}.
In that work, the authors introduce \emph{remorsefree strategies}, which
correspond to strategies which minimize regret in games with $\omega$-regular
objectives. They do not establish lower bounds on the complexity of
realizability or synthesis of remorsefree strategies and they focus on word
strategies of \adam only.

In~\cite{hp06}, Henzinger and Piterman introduce the notion of \emph{good for
games automata}. A non-deterministic automaton is good for solving games if it
fairly simulates the equivalent deterministic automaton.
We show that our notion of regret minimization for word strategies extends this
notion to the quantitative setting (Proposition~\ref{pro:rel-gfg}). Our
definitions give rise to a natural notion of approximate determinization for
weighted automata on infinite words. 

In~\cite{akl10}, Aminof et al. introduce the notion of {\em approximate
determinization by pruning} for weighted sum automata over finite words. For
$\alpha \in (0,1]$, a weighted sum automaton is {\em $\alpha$-determinizable by
pruning} if there exists a finite state strategy to resolve non-determinism and
that constructs a run whose value is at least $\alpha$ times the value of the
maximal run of the given word. So, they consider a notion of approximation which
is a {\em ratio}.  We will show that our concept of regret, when \adam plays word
strategies only, defines instead a notion of approximation with respect to the
{\em difference} metric for weighted automata (Proposition~\ref{pro:rel-dbp}).
There are other differences with their work. First, we consider infinite
words while they consider finite words. Second, we study a general notion of
regret minimization problem in which \eve can use any strategy while they
restrict their study to fixed memory strategies only and leave the problem open
when the memory is not fixed a priori.

Finally, the main difference between these related works and this paper is that
we study the $\inffun$, $\supfun$, $\linffun$, $\lsupfun$,
$\underline{{\sf MP}}$, $\overline{{\sf MP}}$ measures
while they consider the total sum measure or qualitative objectives.

\section{Preliminaries}\label{sec:prelim}
A \emph{weighted arena} is a tuple $G = (V,V_\exists,E,w,v_I)$ where $(V,E,w)$
is a finite edge-weighted graph\footnote{W.l.o.g. $G$ is assumed to be total:
	for each $v \in V$, there exists $v' \in V$ such that $(v,v') \in E$.}
with rational weights, $V_\exists \subseteq V$, and $v_I \in V$ is the initial
vertex.  In the sequel we depict vertices owned by \eve (i.e. $V_\exists$) with
squares and vertices owned by \adam (i.e. $V \setminus V_\exists$) with circles.
We denote the maximum absolute value of a weight in a weighted arena by $W$.

A \emph{play} in a weighted arena is an infinite sequence of vertices $\pi =
v_0v_1 \ldots$ where $v_0 = v_I$ and $(v_i,v_{i+1}) \in E$ for all $i$.
We extend the weight function to partial plays by setting $w(
\langle v_i \rangle_{i = k}^{l}) =
\sum_{i=k}^{l-1} w(v_i,v_{i+1})$.

A \emph{strategy for} \eve (\adam) is a function $\sigma$ that
maps partial plays ending with a vertex $v$ in $V_\exists$ ($V \setminus
V_\exists$) to a successor of $v$. A strategy has memory $m$ if it can be
realized as the output of a finite state machine with $m$ states
(see e.g.~\cite{hpr14} for a formal definition). A \emph{memoryless
(or positional) strategy} is a strategy with memory $1$, that is, a function that
only depends on the last element of the given partial play. A play $\pi = v_0v_1
\ldots$ is \emph{consistent with a strategy} $\sigma$ for \eve (\adam) if
whenever $v_i \in V_\exists$ ($v_i \in V\setminus V_\exists$),
$\sigma(\langle v_j \rangle_{j \le i}) = v_{i+1}$. We denote by $\StrAllE(G)$
($\StrAllA(G)$) the set of all strategies for \eve (\adam) and by $\StrE^{m}(G)$
($\StrA^{m}(G)$) the set of all strategies for \eve (\adam) in $G$ that require
memory of size at most $m$, in particular $\StrPosE(G)$ ($\StrPosA(G)$) is the
set of all memoryless strategies of \eve (\adam) in $G$.  We omit $G$ if the
context is clear.

\paragraph{Payoff functions.}
A play in a weighted arena defines an infinite sequence of weights.  We define
below several classical \emph{payoff functions} that map such sequences to real
numbers.\footnote{The values of all functions are not infinite, and therefore in
	$\mathbb{R}$ since we deal with finite graphs only.}
Formally, for a play $\pi = v_0v_1\ldots$ we define:
\begin{itemize}
	\item the $\inffun$ ($\supfun$) payoff, 
		is the minimum (maximum) weight seen along a play:
\(
	\inffun(\pi) = \inf\{ w(v_i,v_{i+1}) \st i \ge 0\}
\) and \(
	\supfun(\pi) = \sup\{ w(v_i,v_{i+1})\st i \ge 0\};
\)
	\item the $\linffun$ ($\lsupfun$) payoff,
		is the minimum (maximum) weight seen infinitely
		often:
\(
	\linffun(\pi) = \liminf_{i \to \infty} w(v_i,v_{i+1})
\) and, respectively, we have that \(
	\lsupfun(\pi) = \limsup_{i \to \infty} w(v_i,v_{i+1});
\)
	\item the \emph{mean-payoff} value of a play, i.e. the limiting average
		weight, defined using $\liminf$ or $\limsup$ since
		the running averages might not converge:
\(
	\underline{\mpfun}(\pi) = \liminf_{k \to \infty} \frac{1}{k}
	w(\langle v_i \rangle_{i < k}) 
\) and \(
	\overline{\mpfun}(\pi) = \limsup_{k \to \infty} \frac{1}{k}
	w(\langle v_i \rangle_{i < k}). 
\)
In words, $\underline{\mpfun}$ corresponds to the \emph{limit inferior} of the
average weight of increasingly longer prefixes of the play while
$\overline{\mpfun}$ is defined as the \emph{limit superior} of that same
sequence. 
\end{itemize} 

A payoff function $\Val$ is  \emph{prefix-independent} if for all plays $\pi =
v_0v_1 \ldots$, for all $i \geq 0$,  $\Val(\pi) = \Val(\langle v_j \rangle_{j\ge
i})$. It is well-known that $\linffun$, $\lsupfun$, $\underline{\mpfun}$, and
$\overline{\mpfun}$ are prefix-independent.  Often, the arguments that we
develop work uniformly for these four measures because of their
prefix-independent property.  $\inffun$ and $\supfun$ are not prefix-independent
but often in the sequel we apply a simple transformation to the game and encode
$\inffun$ into a $\linffun$ objective, and $\supfun$ into a $\lsupfun$
objective. The transformation consists of encoding in the vertices of the arena
the minimal (maximal) weight that has been witnessed by a play, and label the
edges of the new graph with this same recorded weight.
When this simple transformation does not suffice, we mention it explicitly.

\paragraph{Regret.}
Consider a fixed weighted arena $G$, and payoff function $\Val$.
Given strategies $\sigma, \tau$, for \eve and \adam respectively, and $v \in V$,
we denote by $\out{v}{\sigma}{\tau}$ the unique play starting from $v$ that is
consistent with $\sigma$ and $\tau$ and denote its value by:
\(
	\StratVal{G}{v}{\sigma}{\tau} :=
	\PlayVal{\out{v}{\sigma}{\tau}}.
\)
We omit $G$ if it is clear from the context. If $v$ is omitted, it is
assumed to be $v_I$.

Let $\StrE \subseteq \StrAllE$ and $\StrA \subseteq \StrAllA$ be sets of
strategies for \eve and \adam respectively.  Given $\sigma \in \StrE$ we define
the \emph{regret of $\sigma$ in $G$ w.r.t. $\StrE$ and $\StrA$} as:
\[ \textstyle
	\regret{\sigma}{G}{\StrE,\StrA} := \sup_{\tau \in \StrA} (\sup_{\sigma'
	\in \StrE} \StratVal{}{}{\sigma'}{\tau} - \StratVal{}{}{\sigma}{\tau}).
\]
We define the \emph{regret of $G$ w.r.t. $\StrE$ and $\StrA$} as: 
\[ \textstyle
	\Regret{G}{\StrE,\StrA} := \inf_{\sigma \in \StrE}
		\regret{\sigma}{G}{\StrE,\StrA}.
\]
When $\StrE$ or $\StrA$ are omitted from $\regret{{}}{\cdot}{{}}$ and
$\Regret{\cdot}{{}}$ they are assumed to be the set of all strategies for \eve
and \adam.

\begin{remark}[Ratio vs. difference]\label{rem:metric}
	Let $G$ be a weighted arena and $\StrE \subseteq \StrAllE$ and $\StrA
	\subseteq \StrAllA$. Consider the regret of $G$ defined using the ratio
	measure, instead of difference. For $\inffun$, $\supfun$, $\linffun$, and
	$\lsupfun$, Deciding whether the regret of $G$ is (strictly) less than a
	given threshold $r$ reduces (in polynomial time) to deciding the same
	problem in $G_{\log}$ -- which is obtained by replacing every weight $x$
	in $G$ with $\log_2 x$ -- for threshold $\log_2 r$ with the difference
	measure. 
\end{remark}

We will make use of two other values associated with the vertices of an arena:
the \emph{antagonistic} and \emph{cooperative} values, defined for plays from a
vertex $v \in V$ as
\[ 
	\aVal^v(G) := \sup_{\sigma \in \StrAllE} \inf_{\tau \in \StrAllA}
		\StratVal{}{v}{\sigma}{\tau}
 \qquad
	\cVal^v(G) := \sup_{\sigma \in \StrAllE} \sup_{\tau \in \StrAllA}
		\StratVal{}{v}{\sigma}{\tau}.
\]
When clear from context $G$ will be omitted, and if $v$ is
omitted it is assumed to be $v_I$. 

\begin{remark}
	It is well-known that $\cVal$ and $\aVal$ can be computed in polynomial
	time, w.r.t. the underlying graph of the given arena, for all payoff
	functions but $\mpfun$~\cite{cahs03,cdh10}. For $\mpfun$, $\cVal$ is
	known to be computable in polynomial time, for $\aVal$  it can be done in
	$\UP \cap \coUP$~\cite{jurdzinski98} and in \emph{pseudo-polynomial
	time}~\cite{zp96,bcdgr09}.
\end{remark}

\section{Variant I: \adam plays any strategy}\label{sec:all-strats}
For this variant, we establish that for all the payoff functions that we
consider, the problem of computing the antagonistic value and the problem of
computing the regret value are {\em inter-reducible} in polynomial time. As a
direct consequence, we obtain the following theorem:
\begin{theorem}\label{thm:anyadversary}
	Deciding whether the regret value is less than a given threshold (strictly or
	non-strictly) is $\P$-complete (under log-space reductions) for
	$\inffun$, $\supfun$, $\linffun$, and $\lsupfun$, and equivalent to
	mean-payoff games (under polynomial-time reductions) for
	$\underline{\mpfun}$ and $\overline{\mpfun}$.
\end{theorem}

\subsection{Upper bounds} We now describe an algorithm to compute regret for all
payoff functions. To do so, we will use the fact that all payoff functions we
consider, can be assumed to be prefix-independent. Thus, 
let us first convince the reader that one can, in polynomial time, modify
$\inffun$ and $\supfun$ games so that they become prefix-independent.

\begin{lemma}\label{lem:prefindepping}
	For a given weighted arena $G$ with payoff function $\supfun$:
	\(
		\Regret{G}{} = \Regret{G_{\max}}{};
	\)
	for payoff function $\inffun$:
	\(
		\Regret{G}{} = \Regret{G_{\min}}{}.
	\)
\end{lemma}

Consider a weighted arena $G = (V, V_\exists, v_I, E, w)$. We describe how to
construct $G_{\min}$ from $G$ so that there is a clear bijection between plays
in both games defined with the $\inffun$ payoff function. The arena $G_{\min}$
consists of the following components:
\begin{itemize}
	\item $V' = V \times \{w(e) \st e \in E\}$;
	\item $V'_\exists = \{(v,n) \in V' \st v \in V_\exists\}$;
	\item $v'_I = (v_I, W)$;
	\item $E' \ni \big((u,n),(v,m)\big)$ if and only if $(u,v)
		\in E$ and $m = \min\{n, w(u,v)\}$;
	\item $w'\big((u,n),(v,m)\big) = m$.
\end{itemize}
Intuitively, the construction keeps track of the minimal weight witnessed by a
play by encoding it into the vertices themselves. It is not hard to see that
plays in $G_{\min}$ indeed have a one-to-one correspondence with plays in $G$.
Furthermore, the $\linffun$ and $\lsupfun$ values of a play in $G_{\min}$ are 
easily seen to be equivalent to the $\inffun$ value of the play in $G_{\min}$
and the corresponding play in $G$.

A similar idea can be used to construct weighted arena $G_{\max}$ from a
$\supfun$ game such that the maximal weight is recorded (instead of the
minimal). 

\begin{lemma}\label{lem:regToAll}
	For payoff functions $\inffun$, $\supfun$, $\linffun$, 
	$\lsupfun$, $\underline{\mpfun}$, and
	$\overline{\mpfun}$ computing the regret of a game is at most as
	hard as computing the antagonistic value of a (polynomial-size) game
	with the same payoff function.
\end{lemma}

%
%

We now describe the construction used to prove Lemma~\ref{lem:regToAll}.
Let us fix a weighted arena $G$. We define a new weight function $w'$ as
follows. For any edge $e=(u,v)$ let $w'(e) = -\infty$ if $u \in V
\setminus V_\exists$, and if $u \in V_\exists$ then \( w'(e) = \max
\{\cVal^{v'} \st (u,v') \in E\setminus\{e\}\}.  \) Intuitively, $w'$
represents the best value obtainable for a strategy of \eve that differs
at the given edge. It is not difficult to see that in order to minimize
regret, \eve is trying to minimize the difference between the value given by the
original weight function $w$ and the value given by $w'$.
Let $\range{w'}$ be the set of values $\{w'(e) \st e \in E\}$. For $b \in
\range{w'}$ we define $G^b$ to be the graph obtained by restricting $G$---the
original weighted arena with weight function $w$---to edges $e$ with $w'(e)\leq
b$.

\begin{figure}
\begin{center}
\begin{tikzpicture}
\node[va,initial](v0){$v_0$};
\node[ve,below=of v0](v1){$v_1$};
\node[ve,right=5cm of v1](sink){$v_\bot$};
\node[right=of v1,yshift=1.5cm](vb){$v^{b_1}_I$};
\node[right=of vb](rb){$\cdots$};
\node[right=of v1,yshift=-1.5cm](vb2){$v^{b_n}_I$};
\node[right=of vb2](rb2){$\cdots$};

\node[fit=(vb)(rb)](gb){};
\node[yshift=0.3cm] at (gb.north){$G^{b_1}$};
\node[below=0.4cm of gb](others){$\vdots$};
\node[fit=(vb2)(rb2)](gb2){};
\node[yshift=-0.3cm] at (gb2.south){$G^{b_n}$};

\path
(v0) edge[loop above] node[el,swap] {$0$} (v0)
(v0) edge node[el,swap] {$0$} (v1)
(sink) edge[loop right] node[el,swap] {$-2W-1$} (sink)
(v1) edge node[el] {$0$} (vb)
(v1) edge node[el,swap] {$0$} (vb2)
(vb) edge node[el] {$w(e)-b_1$} (rb)
(vb2) edge node[el] {$w(e)-b_n$} (rb2)
(gb) edge[dotted] node[el] {if $w'(e) > b_1$} node[el,swap,pos=0.3] {$0$} (sink)
(gb2) edge[dotted] node[el,swap] {if $w'(e) > b_n$} node[el,pos=0.3] {$0$} (sink)
;

\end{tikzpicture}
\caption{Weighted arena $\widehat{G}$, constructed from $G$. Dotted lines represent
several edges added when the condition labelling it is met.}
\label{fig:ghat-algo}
\end{center}
\end{figure}

Next, we will construct a new weighted arena $\widehat{G}$ such that the
regret of $G$ is a function of the {\em antagonistic} value of $\widehat{G}$.
Figure~\ref{fig:ghat-algo} depicts the general form of the arena we
construct. We have three vertices $v_0 \in \widehat{V} \setminus
\widehat{V_\exists}$ and $v_1,v_\bot \in \widehat{V_\exists}$ and a ``copy'' of
$G$ as $G^b$ for each $b \in \range{w'} \setminus \{-\infty\}$. We have
a self-loop of weight $0$ on $v_0$ which is the initial vertex of
$\widehat{G}$, a self-loop of weight $-2W-1$ on $v_\bot$, and weight-$0$
edges from $v_0$ to $v_1$ and from $v_1$ to the initial vertices of
$G^b$ for all $b$.  Recall that $G^b$ might not be total. To fix this we
add, for all vertices without a successor, a weight-$0$ edge to
$v_\bot$. The remainder of the weight function $\widehat{w}$, is defined for
each edge $e^b$ in $G^b$ as $\widehat{w}(e^b) = w(e)-b$. 

Intuitively, in $\widehat{G}$ \adam first decides whether he can ensure a
non-zero regret. If this is the case, then he moves to $v_1$. Next, \eve
chooses a maximal value she will allow for strategies which differ from
the one she will play (this is the choice of $b$). The play then moves
to the corresponding copy of $G$, i.e. $G^b$. She can now play to
maximize her mean-payoff value. However, if her choice of $b$ was not
correct then the play will end in $v_\bot$.

We show that, for all prefix-independent payoff functions we consider, the
following holds:
\begin{claim}
	For all prefix-independent payoff functions considered in this work
\(\Regret{G}{} = -\aVal(\widehat{G}).\)
\end{claim}
This implies Lemma~\ref{lem:regToAll} for all prefix-independent payoff
functions. Together with Lemma~\ref{lem:prefindepping}, we get the same result
for $\inffun$ and $\supfun$.

\begin{proof}[Proof of the Claim]
Let us start by arguing that the following equality holds.
\begin{equation}\label{eqn:devRegret}
	\Regret{G}{} = \inf_{\sigma \in \StrAllE} \sup_{\tau \in \StrAllA}
	\sup_{\sigma' \in \StrAllE \setminus\{\sigma\}}
	\{0,\StratVal{}{}{\sigma'}{\tau} - \StratVal{}{}{\sigma}{\tau} \}.
\end{equation}
Indeed, it follows from the definition of regret that if $\sigma' = \sigma$ then
the regret of the game is $0$. Thus, \adam can always ensure the regret of a
game is at least $0$. Now, for $b \in \range{w'}$, define $\StrE(b) \subseteq
\StrAllE(G)$ as:
\[
	\StrE(b) := \{ \sigma \st \sup_{\tau \in \StrAllA} \sup_{\sigma' \in
	\StrAllE \setminus \{\sigma\}} \StratVal{}{}{\sigma'}{\tau} \leq b\}.
\] 
It is clear from the definitions that $\sigma\in \StrE(b)$ if and only if
$\sigma$ is a strategy for \eve in $G^b$ which avoids ever reaching $v_\bot$.
Now, if we let 
\[
	b_\sigma = \sup_{\tau \in \StrAllA} \sup_{\sigma' \in \StrAllE
		\setminus\{\sigma\}} \StratVal{}{}{\sigma'}{\tau},
\]
then $\sigma \in \StrE(b)$ if and only $b_\sigma \leq b$.  It follows that for
all $\sigma$:
\begin{equation}\label{eqn:devInf}
	\sup_{\tau \in \StrAllA} \sup_{\sigma' \in \StrAllE \setminus \{\sigma\}}
	\StratVal{}{}{\sigma'}{\tau} = \inf \{ b \st \sigma \in \StrE(b)\}.
\end{equation}
We now turn to the mean-payoff game played on $\widehat{G}$, and make some
observations about the strategies we need to consider. It is well known that
memoryless strategies suffice for either player to ensure an antagonistic value
of at least (resp. at most) $\aVal(\widehat{G})$, for all quantitative games
considered in this work, so we can assume that \adam and \eve play positionally.
It follows that all plays either remain in $v_0$, or move to $G^b$ for some $b$,
and \adam can ensure a non-positive payoff.  Note that for $b_{\max} =
\max(\range{w'} \setminus \{-\infty\})$ we have $G^{b_{\max}} = G$.  So the copy
of $G^{b_{\max}}$ in $\widehat{G}$ has no edge to $v_\bot$, and by playing to this
sub-graph \eve can ensure a payoff of at least $-|b_{\max} - W| \geq -2W$.  As
any play that reaches $v_\bot$ will have a payoff of $-2W-1$, we can restrict
\eve to strategies which avoid $v_\bot$, and hence all plays either remain in
$v_0$ or (eventually) in the copy of $G^b$ for some $b$.  Now $G^b$ contains no
restrictions for \adam, so we can assume that he plays the same strategy in all
the copies of $G^b$ (where he cannot force the play to $v_\bot$), and these
strategies have a one-to-one correspondence with strategies in $G$.  Likewise,
as \eve chooses a unique $G^b$ to play in, we have a one-to-one correspondence
with strategies of \eve in $\widehat{G}$ and strategies in $G$. More precisely, if
$\widehat{\sigma} \in \StrAllE(\widehat{G})$ is such that $\widehat{\sigma}(v_1) = v_I^b$
and $\widehat{\sigma}$ avoids $v_\bot$, then the corresponding strategy $\sigma \in
\StrAllE(G)$ is a valid strategy in $G^b$, and hence:
\begin{equation}\label{eqn:stratEquiv}
	\widehat{\sigma}(v_1)=v_I^b \implies \sigma \in \Sigma(b).
\end{equation}
Now suppose $\widehat{\sigma} \in \StrAllE(\widehat{G})$ is a strategy such that
$\widehat{\sigma}(v_1) = v_I^b$ and $\widehat{\sigma}$ avoids $v_\bot$, and $\widehat{\tau}
\in \StrAllA(\widehat{G})$ is a strategy such that $\widehat{\tau}(v_0) = v_1$.  Let $\sigma
\in \Sigma(b)$ and $\tau \in \StrAllA(G)$ be the strategies in $G$ corresponding
to $\widehat{\sigma}$ and $\widehat{\tau}$ respectively.  
It is easy to show that:
\begin{equation}\label{eqn:GtoGhat}
	-\StratVal{\widehat{G}}{ }{\widehat{\sigma}}{\widehat{\tau}} =
	b - \StratVal{G}{ }{\sigma}{\tau}.
\end{equation}
Putting together Equations~\eqref{eqn:devRegret}--\eqref{eqn:GtoGhat} gives:
\[
\begin{array}{rclr}
	-\aVal(\widehat{G}) &=& - \sup_{\widehat{\sigma}} \inf_{\widehat{\tau}}
		\StratVal{\widehat{G}}{}{\widehat{\sigma}}{\widehat{\tau}}\\[1.5ex]
	&=& \inf_{\widehat{\sigma}} \sup
		(\{-\StratVal{\widehat{G}}{}{\widehat{\sigma}}{\widehat{\tau}}\st
		\widehat{\tau}(v_0) = v_1 \} \cup \{0\})\\[1.5ex]
	&=& \inf \{ \sup (\{-\StratVal{\widehat{G}}{}{\widehat{\sigma}}{\widehat{\tau}}\st
		\widehat{\tau}(v_0) = v_1 \} \cup \{0\}) \st \widehat{\sigma}(v_1) =
		v_I^b \}\\[1.5ex]
	&=&  \inf \{ \sup_{\tau \in \StrAllA} (\{b-\StratVal{G}{}{\sigma}{\tau} \}
		\cup \{0\}) \st \sigma \in \Sigma(b)\}\\[1.5ex]
	&=& \inf_{\sigma \in \StrAllE} \sup_{\tau \in \StrAllA} (\{\inf\{b\st
		\sigma \in \Sigma(b)\} -\StratVal{G}{}{\sigma}{\tau} \} \cup
		\{0\})\\[1.5ex]
	&=& \inf_{\sigma \in \StrAllE} \sup_{\tau \in \StrAllA} \sup_{\sigma'
		\in \StrAllE} \{0,\StratVal{G}{ }{\sigma'}{\tau} -
		\StratVal{G}{}{\sigma}{\tau} \}\\[1.5ex]
	&=& \Regret{G}{}\text{ as required.} & {}
\end{array}
\]
\end{proof}

\subsection{Lower bounds}  For all the payoff functions, from 
$G$ we can construct in logarithmic space $G'$ such that the antagonistic
value of $G$ is a function of the regret value of $G'$, and so we have:

\begin{lemma}\label{lem:AlltoRegret}
	For payoff functions $\inffun$, $\supfun$, $\linffun$, 
	$\lsupfun$, $\underline{\mpfun}$, and
	$\overline{\mpfun}$ computing the regret of a game is at least as
	hard as computing the antagonistic value of a (polynomial-size) game
	with the same payoff function.
\end{lemma}

\begin{figure}
\begin{center}
\begin{tikzpicture}[inner sep=2mm, ve/.style={rectangle,
	draw},va/.style={circle, draw}, node distance=1cm]
\node[ve,initial](A){{$v_I'$}};
\node[ve,dotted,below=of A,label=below:{arena $G$}](B){{$v_I$}};
\node[va,right=of A](C){};
\node[va,right=of C, yshift=.5cm](D){};
\node[va,right=of C, yshift=-.5cm](E){};

\path
(A) edge node[el]{$L$} (B)
(A) edge node[el]{$L$} (C)
(C) edge node[el]{$M_1$}(D)
(C) edge node[el,swap]{$M_2$}(E)
(D) edge[loop above] node[el,swap]{$N_1$} (D)
(E) edge[loop below] node[el]{$N_2$} (E);
\end{tikzpicture}
\caption{Gadget to reduce a game to its regret game.}\label{fig:AlltoRegret}
\end{center}
\end{figure}

\begin{proof}
	Suppose $G$ is a weighted arena with initial vertex $v_I$. Consider the
	weighted arena $G'$ obtained by adding to $G$ the gadget of
	Figure~\ref{fig:AlltoRegret}.  The initial vertex of $G'$ is set to be
	$v'_I$. In $G'$ from $v'_I$ \eve can either progress to the original
	game or to the new gadget, both with weight $L$.	
	We claim that the right choice of values for the parameters
	$L,M_1,M_2,N_1,N_2$ makes it so that the antagonistic value of $G$ is a
	function of the regret of the game $G'$. 

	Let us first give the values of  $L$, $M_1$, $M_2$, $N_1$, and $N_2$ for
	each of the payoff functions considered.  For all our payoff functions
	we have $M_1=M_2=L$; $N_1 = W+1$; and $N_2 = -3W-2$.  For $\inffun$ we
	have $L=W$, for $\supfun$ we have $L=-W$ and for the remaining payoff
	functions we have $L=0$.

	The following result shows that computing the regret of $G$ is at least
	as hard as computing the (antagonistic) value of $G'$.
	\begin{claim}
		For all payoff functions:
		\[\aVal(G) = W+1-\Regret{G'}{}.\]
	\end{claim}
	We first observe that for all payoff functions we consider we have that
	$\aVal(G)$ and $\cVal(G)$ both lie in $[-W,W]$.
	At $v_I'$ \eve has a choice: she can choose to remain in the gadget or
	she can move to the original game $G$.  If she chooses to remain in the
	gadget, her payoff will be $-3W-2$, meanwhile \adam could choose a
	strategy that would have achieved a payoff of $\cVal(G)$ if she had
	chosen to play to $G$.  Hence her regret in this case is $\cVal(G)+3W+2
	\geq 2W+2$.  Otherwise, if she chooses to play to $G$ she can achieve a
	payoff of at most $\aVal(G)$.  As $\cVal(G) \leq W$ is the maximum
	possible payoff achievable in $G$, the strategy which now maximizes
	\eve's regret is the one which remains in the gadget -- giving a payoff
	of $W+1$. Her regret in this case is $W+1 - \aVal(G) \leq 2W + 1$.
	Therefore, to minimize her regret she will play this strategy, and
	$\Regret{G'}{} = W+1-\aVal(G)$.
\end{proof}

\subsection{Memory requirements for \eve and \adam} It follows from the
reductions underlying the proof of Lemma~\ref{lem:regToAll} that \eve only
requires positional strategies to minimize regret when there is no restriction
on \adam's strategies. On the other hand, for any given strategy $\sigma$ for
\eve, the strategy $\tau$ for \adam which witnesses the maximal regret against it
consists of a combination of three positional strategies: first he moves to the
optimal vertex for deviating (it is from this vertex that the alternative
strategy $\sigma'$ of \eve will achieve a better payoff against $\tau$), then he plays his
optimal (positional) strategy in the antagonistic game (i.e. against $\sigma$).
His strategy for the alternative scenario, i.e. against $\sigma'$, is his
optimal strategy in the co-operative game which is also positional. This
combined strategy is clearly realizable as a strategy with three memory states,
giving us:
\begin{corollary}
	For payoff functions $\linffun$, $\lsupfun$, $\underline{\mpfun}$ and
	$\overline{\mpfun}$: \( \Regret{G}{} =
	\Regret{G}{\StrPosE,\Sigma^3_\forall}.\) 
\end{corollary}
The algorithm we give relies on the prefix-independence of the payoff function.
As the transformation from $\inffun$ and $\supfun$ to equivalent
prefix-independent ones is polynomial it follows that polynomial memory (w.r.t.
the size of the underlying graph of the arena) suffices for both players.

\section{Variant II: \adam plays memoryless strategies}
\label{sec:pos-strats}
For this variant, we provide a polynomial space algorithm to
solve the problem for all the payoff functions, we then provide lower bounds.
\begin{theorem}\label{thm:memlessadversary}
	Deciding whether the regret value is less than a given threshold (strictly or
	non-strictly) playing against memoryless strategies of \adam is
	\PSPACE-complete for $\linffun$, $\underline{\mpfun}$ and
	$\overline{\mpfun}$; in \PSPACE~and \coNP-hard for $\inffun$, $\supfun$
	and $\lsupfun$.
\end{theorem}

\subsection{Upper bounds}
Let us now show how to compute regret against positional adversaries.

\begin{lemma}
	\label{lem:pspace-memless-regret}
	For payoff functions $\inffun$, $\supfun$, $\linffun$, $\lsupfun$,
	$\underline{\mpfun}$ and $\overline{\mpfun}$, the regret of a game
	played against a positional adversary can be computed in polynomial
	space.
\end{lemma}
Given a weighted arena $G$, we construct
a new weighted arena $\widehat{G}$ such that we have 
that $-\aVal(\widehat{G})$ is equivalent to the regret of $G$. 

\begin{figure}
\begin{center}
\begin{tikzpicture}
\node[ve,initial](A){{$u$}};
\node[ve,below=of A](B){{$v$}};
\node[va,right=of B](C){{$x$}};

\path
(A) edge[loop above] node[el,swap]{$1$} (A)
(A) edge[bend left] node [el,swap]{$0$} (B)
(B) edge[bend left] node [el]{$0$} (A)
(B) edge[bend left] node [el,swap]{$0$} (C)
(C) edge[bend left] node [el]{$0$} (B)
(C) edge node [el,swap]{$6$} (A);

\end{tikzpicture}
\caption{Example weighted arena $G_1$.}\label{fig:she-needs-mem}
\end{center}
\end{figure}

The vertices of $\widehat{G}$ encode the choices made by \adam. For a
subset of edges $D \subseteq E$, let $G \restriction D$ denote the
weighted arena $(V, V_\exists, D, w, v_I)$.  The new weighted
arena $\widehat{G}$ is the tuple $(\widehat{V}, \widehat{V_\exists}, \widehat{E},
\widehat{w}, \widehat{v_I})$ where
\begin{inparaenum}[$(i)$]
	\item $\widehat{V} = V \times \pow(E)$;
	\item $\widehat{V_\exists} = \{(v,D) \in \widehat{V} \st v \in V_\exists \}$;
	\item $\widehat{v_I} = (v_I, E)$;
	\item $\widehat{E}$ contains the edge $\big( (u,C) , (v,D) \big)$ if
		and only if $(u,v) \in C$ and, either $u \in \VtcE$ and
		$D = C$, or $u \in \VtcA$ and $D = C \setminus \{ (u, x)
		\in E \st x \neq v\}$;
	\item $\widehat{w}\big( (u,C), (v,D) \big) = w(u,v) - \cVal(G
		\restriction D)$.
\end{inparaenum}
The application of this transformation for the graph of
Fig.~\ref{fig:she-needs-mem} w.r.t. to the $\underline{\mpfun}$ payoff
function is given in Fig.~\ref{fig:eg-posalgo}.

\begin{figure}
\begin{center}
\begin{tikzpicture}
\node[ve,initial](A0) at (0,4) {{$u, \{xv,xu\}$}};
\node[ve](B0) at (0,2) {{$v, \{xv,xu\}$}};
\node[ellipse,draw](C0) at (0,0) {{$x, \{xv,xu\}$}};
\node[ve](A1) at (4,4) {{$u, \{xu\}$}};
\node[ve](B1) at (7,4) {{$v, \{xu\}$}};
\node[ellipse,draw](C1) at (7,2) {{$x, \{xu\}$}};
\node[ve](A2) at (7,0) {{$u, \{xv\}$}};
\node[ve](B2) at (4,0) {{$v, \{xv\}$}};
\node[ellipse,draw](C2) at (4,2) {{$x, \{xv\}$}};

\path
(A0) edge[loop above] node[el]{$-1$} (A0)
(A0) edge[bend left] node [el]{$-2$} (B0)
(B0) edge[bend left] node [el]{$-2$} (A0)
(B0) edge node [el]{$-2$} (C0)
(C0) edge node [el]{$4$} (A1)
(C0) edge node [el,swap]{$-1$} (B2)
(A1) edge[loop above] node [el]{$-1$} (A1)
(B1) edge node [el]{$-2$} (C1)
(B1) edge[bend left] node [el,swap]{$-2$} (A1)
(A1) edge[bend left] node [el]{$-2$} (B1)
(C1) edge node [el]{$4$} (A1)
(A2) edge[bend left] node [el]{$-1$} (B2)
(B2) edge[bend left] node [el,swap]{$-1$} (A2)
(C2) edge[bend left,pos=0.4] node [el]{$-1$} (B2)
(B2) edge[bend left,pos=0.6] node [el]{$-1$} (C2)
(A2) edge[loop above] node[el]{$0$} (A2)
;
\end{tikzpicture}
\caption{Weighted arena $\widehat{G_1}$, constructed from $G_1$ w.r.t the
$\underline{\mpfun}$ payoff function. In the edge set
component only edges leaving \adam nodes are depicted.}
\label{fig:eg-posalgo}
\end{center}
\end{figure}

Consider a play $\widehat{\pi} = (v_0,C_0) (v_1, C_1) \ldots$ in $\widehat{G}$. We
denote by $\proj{\widehat{\pi}}{k}$, for $k \in \{1,2\}$, the sequence $\langle
c_{k,i} \rangle_{i \ge 0}$, where $c_{k,i}$ is the $k$-th component of the
$i$-th pair from $\widehat{\pi}$. Observe that $\proj{\widehat{\pi}}{1}$ is a valid play in
$G$. Also observe that $E \supseteq C_j \supseteq C_{j+1}$ for all $j$.  Hence
$\proj{\widehat{\pi}}{2}$ is an infinite descending chain of finite subsets, and
therefore $\lim\proj{\widehat{\pi}}{2}$ is well-defined.  Finally, we define
\(
	\cost{\widehat{\pi}} := \cVal(G \restriction \lim \proj{\widehat{\pi}}{2}).
\)
The following result relates the value of a play in $\widehat{G}$ to the value of
the corresponding play in $G$.

\begin{lemma}
	\label{lem:rel-hat-plays}
	For payoff functions
	$\linffun,\lsupfun,\underline{\mpfun},\overline{\mpfun}$ and for any
	play $\widehat{\pi}$ in $\widehat{G}$ we have that
	\( 
		\PlayVal{\widehat{\pi}} = \PlayVal{\proj{\widehat{\pi}}{1}} -
			\cost{\widehat{\pi}}.
	\)
\end{lemma}
\begin{proof}
	We first establish the following intermediate result.  It follows from the
	existence of $\lim\proj{\widehat{\pi}}{2}$ and the definition of $\cost{\cdot}$
	that:
	\begin{equation}
		\label{eqn:cost}
		\limsup_{n \to \infty} \frac{1}{n} \sum_{i = 0}^{n-1}
		\cVal(G \restriction C_i) =
		\liminf_{n \to \infty} \frac{1}{n} \sum_{i = 0}^{n-1}
		\cVal(G \restriction C_i) =
		\cost{\widehat{\pi}}.
	\end{equation}
	We now show that the result holds for $\underline{\mpfun}$.
	\begin{align*}
		\PlayVal{\widehat{\pi}} &= \liminf_{n \to \infty}
		\left(\frac{1}{n} \sum_{i=0}^{n-1} \left(w(v_i, v_{i+1}) -
		\cVal(G\restriction C_j)\right) \right)
		& \mbox{defs. of } \PlayVal{\cdot},\widehat{w} \\
		&= \PlayVal{\proj{\widehat{\pi}}{1}} - \limsup_{n \to \infty}
		\frac{1}{n} \sum_{j=0}^{n-1} \cVal(G\restriction C_j) &
		\mbox{def. of } \PlayVal{\cdot} \\
		&= \PlayVal{\proj{\widehat{\pi}}{1}} - \cost{\widehat{\pi}} &
		\mbox{from Eq. (\ref{eqn:cost})}
	\end{align*}
		
	The proofs for the other payoff functions are almost identical (for
	$\linffun$ and $\lsupfun$ replace the use of Equation~\eqref{eqn:cost} by
	Equation~\eqref{eqn:cost2}).
	\begin{equation}
		\label{eqn:cost2}
		\limsup_{i \to \infty} \cVal(G \restriction C_i) =
		\liminf_{i \to \infty} \cVal(G \restriction C_i) =
		\cost{\widehat{\pi}}.
	\end{equation}
\end{proof}

We now describe how to translate winning strategies for either player from
$\widehat{G}$ back to $G$, i.e. given an optimal maximizing (minimizing)
strategy for \eve (\adam) in $\widehat{G}$ we construct the corresponding
optimal regret minimizing strategy (memoryless regret maximizing
counter-strategy) for \eve (\adam) in $G$. For clarity, we follow this
same naming convention throughout this section: again, we say a strategy is an
optimal maximizing (minimizing) strategy when we speak about antagonistic
and cooperative games, we say a strategy is an optimal regret maximizing 
(regret minimizing) when we speak about regret games. When this does not suffice,
we explicitly state which kind of game we are speaking about.

Let $\widehat{\epsilon} \in \StrAllE(\widehat{G})$ be an optimal maximizing strategy of
\eve in $\widehat{G}$ and $\widehat{\alpha} \in \StrAllA(\widehat{G})$ be an optimal
minimizing strategy of \adam.  Indeed, in~\cite{em79} it was shown that
mean-payoff games are positionally determined. We will now define a strategy for
\eve in $G$ which for every play prefix $s$ constructs a valid play prefix
$\widehat{s}$ in $\widehat{G}$ and plays as $\widehat{\epsilon}$ would in $\widehat{G}$ for
$\widehat{s}$. More formally, for a play prefix $s$ from $G$, denote by
$\projinv{s}$ the corresponding sequence of vertex and edge-set pairs in
$\widehat{G}$ (indeed, it is the inverse function of $\proj{\cdot}{1}$, which is
easily seen to be bijective).
Define $\sigma \in \StrAllE(G)$
as follows: $\sigma(s) = \proj{\widehat{\epsilon}(\projinv{s})}{1}$ for all play
prefixes $s \in V^* \cdot \VtcE$ in $G$ consistent with a positional strategy of
\adam.

For a fixed strategy of \eve we can translate the optimal minimizing strategy of
\adam in $\widehat{G}$ into an optimal memoryless regret maximizing counter-strategy
of his in $G$. Formally, for an arbitrary strategy $\sigma$ for \eve in $G$,
define $\widehat{\sigma} \in \StrAllE(\widehat{G})$ as follows: $\widehat{\sigma}(\widehat{s})
= \sigma(\proj{\widehat{s}}{1})$ for all $\widehat{s} \in \widehat{V}^* \cdot
\widehat{\VtcE}$.  Let $\tau_\sigma$ be an optimal (positional) maximizing strategy
for \adam in $G \restriction \lim\proj{\out{}{\widehat{\sigma}}{\widehat{\alpha}}}{2}$.

It is not hard to see the described strategy of \eve ensures a
regret value of at most $-\aVal(\widehat{G})$. Slightly less obvious is the fact
that for any strategy of \eve, the counter-strategy $\tau_\sigma$ of \adam is
such that
\(
	\sup_{\sigma' \in \StrAllE}
	\StratVal{G}{}{\sigma'}{\tau_\sigma} -
	\StratVal{G}{}{\sigma}{\tau_\sigma} \ge -\aVal(\widehat{G}).
\)

\begin{lemma}
	\label{lem:reg-val-hat}
	For payoff functions 
	$\linffun$, $\lsupfun$, $\underline{\mpfun}$, and $\overline{\mpfun}$:
	\[
		\Regret{G}{\StrAllE,\StrPosA} = - \aVal(\widehat{G}). 
	\]
\end{lemma}
\begin{proof}
	The proof is decomposed into two parts. First, we describe a strategy
	$\sigma \in \StrAllE(G)$ which ensures a regret value of at most
	$-\aVal(\widehat{G})$. Second, we show that for any $\sigma \in \StrAllE(G)$
	there is a $\tau \in \StrPosA(G)$ such that
	\[
		\sup_{\sigma' \in \StrAllE} \StratVal{G}{}{\sigma'}{\tau} -
		\StratVal{G}{}{\sigma}{\tau} \ge -\aVal(\widehat{G}).
	\]
	The result follows.

	We have already mentioned earlier that for a play $\widehat{\pi}$ in
	$\widehat{G}$ we have that $\proj{\widehat{\pi}}{1}$ is a play in $G$. Let
	$\mathrm{PPref}(G)$ denote the set of all play prefixes consistent with
	a positional strategy for \adam in $G$.	It is not
	difficult to see that $\proj{\cdot}{1}$ is indeed a bijection between
	plays of $\widehat{G}$ and plays of $G$ consistent with positional
	strategies for \adam.

	It follows from the determinacy of antagonistic games defined by the
	payoff functions considered in this work
	that there are optimal strategies for \eve and \adam that ensure a
	payoff of, respectively, at least and at most a value
	$\aVal(\widehat{G})$ against any strategy of the opposing player. Let
	$\widehat{\epsilon} \in \StrAllE(\widehat{G})$ be an optimal maximizing strategy
	of \eve in $\widehat{G}$ and $\widehat{\alpha} \in \StrAllA(\widehat{G})$ be an
	optimal minimizing strategy of \adam.

	\item \paragraph{(First Part).}
	Define a strategy $\sigma$ from $\StrAllE(G)$ as follows: $\sigma(s) =
	\proj{\widehat{\epsilon}(\projinv{s})}{1}$ for all $s \in \mathrm{PPref}(G) \cdot
	\VtcE$. We claim that
	\[
		\regret{\sigma}{G}{\StrAllE,\StrPosA} \le - \aVal(\widehat{G}).
	\]
	Towards a contradiction, assume there are $\tau \in \StrPosA(G)$ and
	$\sigma' \in \StrAllE(G)$ such that 
	\[
		\StratVal{G}{}{\sigma'}{\tau} - \StratVal{G}{}{\sigma}{\tau} >
		-\aVal(\widehat{G}).
	\]
	Define a strategy $\widehat{\tau} \in \StrAllA(\widehat{G})$ as follows:
	$\widehat{\tau}(\widehat{s}) = \tau(\proj{\widehat{s}}{1})$ for all $\widehat{s} \in
	\widehat{V}^* \cdot \widehat{\VtcA}$. From the definition of $\widehat{\epsilon}$ and
	our assumption we get that
	\begin{equation}
		\label{eqn:eq1}
		\StratVal{G}{}{\sigma'}{\tau} -
		\StratVal{G}{}{\sigma}{\tau} > -\aVal(\widehat{G}) \ge
		-\StratVal{\widehat{G}}{}{\widehat{\epsilon}}{\widehat{\tau}}.
	\end{equation}
	It is straightforward to verify that $\projinv{\out{ }{\sigma}{\tau}} =
	\out{ }{\widehat{\epsilon}}{\widehat{\tau}}$. Therefore, from
	Lemma~\ref{lem:rel-hat-plays}, we have:
	\begin{equation}
		\label{eqn:contra1}
		\PlayVal{\out{ }{\sigma'}{\tau}} > \PlayVal{\out{
		}{\sigma}{\tau}} - \PlayVal{\out{ }{\widehat{\epsilon}}{\widehat{\tau}}} =
		\cVal(G \restriction \lim \proj{\out{
		}{\widehat{\epsilon}}{\widehat{\tau}}}{2}).
	\end{equation}
	At this point we note that, since $\tau$ is a positional strategy, it
	holds that $\StratVal{G}{}{\sigma'}{\tau}$ is at most the highest
	payoff value attainable in $G$ restricted to the edges allowed by
	$\tau$. Formally, if $E_\tau = \{(u,v) \in E \st u \in \VtcA \implies
	v = \tau(u)\}$ then $\PlayVal{\out{ }{\sigma'}{\tau}} \le \cVal(G
	\restriction E_\tau)$. Also, by construction of $\widehat{\tau}$ we get that
	$E_\tau \subseteq \lim \proj{\out{ }{\widehat{\epsilon}}{\widehat{\tau}}}{2}$. It
	should be clear that this implies $\cVal(G \restriction \lim \proj{\out{
	}{\widehat{\epsilon}}{\widehat{\tau}}}{2}) \ge \cVal(G \restriction E_\tau)$. This
	contradicts Equation~\eqref{eqn:contra1}.

	\item \paragraph{(Second Part).}
	For the second part of the proof we require the following result which
	relates positional strategies for \adam in $G$ that agree on certain
	vertices to strategies in sub-graphs defined by plays in $\widehat{G}$.

	\item \begin{claim}\label{clm:obsEquiv}
	Let $\sigma \in \StrAllE(G)$ and $\tau, \tau' \in \StrPosA(G)$.  Then
	$\out{}{\sigma}{\tau} = \out{}{\sigma}{\tau'}$ if and only if $\tau' \in
	\StrPosA(G \restriction \lim \proj{\projinv{\out{}{\tau}{\sigma}}}{2})$
	\end{claim}
	\begin{proof}
		\textit{(only if)} Note that by construction of $\widehat{G}$ we
		have that once \adam chooses an edge $\big((u,C), (v,D)\big)$
		from a vertex $(v, C) \in \widehatVtcA$ then on any subsequent
		visit to a vertex $(u, C') \in \widehatVtcA$ he has no other
		option but to go to $(v, C')$.  That is, his choice is
		restricted to be consistent with the history of the play. For a
		play $\widehat{\pi}$ in $\widehat{G}$, it is clear that the sequence
		$\proj{\widehat{\pi}}{2}$ is the decreasing sequence of sets of
		edges consistent (for \adam) with the history of the play in the
		same manner.  In particular, for any $\tau' \in \StrPosA(G)$ and
		any play $\pi$ in $G$ consistent with $\tau'$ we have that
		$\tau'$ is a valid strategy for \adam in $G \restriction E'$
		where $E' = \lim \proj{\projinv{\pi}}{2}$.  As
		$\out{}{\sigma}{\tau} = \out{}{\sigma}{\tau'}$ is a play
		consistent with $\tau'$, the result follows.

		\textit{(if)} Suppose $\out{}{\sigma}{\tau} \neq
		\out{}{\sigma}{\tau'}$, and let $v$ be the last vertex in their
		common prefix.  As $\sigma$ is common to both plays, we have $v
		\in \VtcA$, and $\tau(v) \neq \tau'(v)$.  In particular,
		$(v,\tau'(v)) \notin \lim
		\proj{\projinv{\out{}{\tau}{\sigma}}}{2}$ so $\tau' \notin
		\StrPosA(G \restriction \lim
		\proj{\projinv{\out{}{\tau}{\sigma}}}{2})$.
	\end{proof}

	For an arbitrary strategy $\sigma$ for \eve in $G$, define $\widehat{\sigma}
	\in \StrAllE(\widehat{G})$ as follows: $\widehat{\sigma}(\widehat{s}\cdot(v,D)) =
	(\sigma(\projinv{\widehat{s} \cdot (v,D)}),D)$ for all $\widehat{s}\cdot(v,D)
	\in \widehat{V}^* \cdot \widehat{\VtcE}$.
	Let $\tau_\sigma$ be an optimal
	(positional) maximizing strategy for \adam in $G \restriction
	\lim\proj{\out{}{\widehat{\sigma}}{\widehat{\alpha}}}{2}$. We claim that for all
	$\sigma \in \StrAllE(G)$ we have that
	\[
		\sup_{\sigma' \in \StrAllE} \StratVal{G}{}
		{\sigma'}{\tau_\sigma} - \StratVal{G}{}{\sigma}{\tau_\sigma} \ge -
		\aVal(\widehat{G}).
	\]
	Towards a contradiction, assume that for some $\sigma \in \StrAllE(G)$
	it is the case that for all $\sigma' \in \StrAllE(G)$ the left hand side
	of the above inequality is strictly smaller than the right hand side. By
	definition of $\widehat{\alpha}$ we then get the following inequality.
	\begin{equation}
		\label{eqn:ineq-contra}
		\sup_{\sigma' \in \StrAllE} \StratVal{G}{}
		{\sigma'}{\tau_\sigma} -
		\StratVal{G}{}{\sigma}{\tau_\sigma} < - \aVal(\widehat{G}) \le -
		\StratVal{\widehat{G}}{}{\widehat{\sigma}}{\widehat{\alpha}}
	\end{equation}
	Using the above Claim it is easy to show that $\proj{\out{}
	{\sigma}{\tau_\sigma}}{1} = \out{}{\widehat{\sigma}}{\widehat{\alpha}}$.
	Hence, by Equation~\eqref{eqn:ineq-contra} and
	Lemma~\ref{lem:rel-hat-plays} we get that:
	\begin{equation}
		\label{eqn:eq2}
		\sup_{\sigma' \in \StrAllE} \PlayVal{\out{
		}{\sigma'}{\tau_\sigma}} < \cVal(G \restriction \lim \proj{\out{
		}{\widehat{\sigma}}{\widehat{\alpha}}}{2})
	\end{equation}
	However, by choice of $\tau_\sigma$, we know that there is a strategy
	$\sigma'' \in \StrAllE(G)$ such that $\PlayVal{\out{
	}{\sigma''}{\tau_\sigma}} = \cVal(G \restriction \lim \proj{\out{
	}{\widehat{\sigma}}{\widehat{\alpha}}}{2})$. This contradicts
	Equation~\eqref{eqn:eq2} and completes the proof of the Theorem.
\end{proof}

If $G$ was constructed from a $\inffun$ or $\supfun$ game $H$, then
one could easily transfer the described strategy of \eve, $\sigma$ into a
strategy for her in $H$ which achieves the same regret. In order to have a
symmetric result we still lack the ability to transfer a strategy of \adam from
$\widehat{G}$ to the original game $H$. Consider a modified construction in which we
additionally keep track of the minimal (resp. maximal) weight seen so far by a
play, just like described in Sect.~\ref{sec:all-strats}. Denote the
corresponding game by $\tilde{G}$. The vertex set $\tilde{V}$ of $\tilde{G}$ is
thus a set of triples of the form $(v,C,x)$ where $x$ is the minimal (resp.
maximal) weight the play has witnessed. We observe that in the proof of the
above result the intuition behind why we can transfer a strategy of \adam from
$\widehat{G}$ back to $G$ as a memoryless strategy, \textbf{although the vertices in
$\widehat{G}$ already encode additional information}, is that once we have fixed a
strategy of \eve in $G$, this gives us enough information about the prefix of
the play before visiting any \adam vertex. In other words, we construct a
strategy of \adam tailored to spoil a specific strategy of \eve, $\sigma$, in
$G$ using the information we gather from $\projinv{\cdot}$ and his optimal
strategy in $\widehat{G}$. These properties still hold in $\tilde{G}$. Thus, we get
the following result.
\begin{lemma}
	For payoff functions 
	$\inffun$, $\supfun$:
	\(
		\Regret{G}{\StrAllE,\StrPosA} = - \aVal(\tilde{G}). 
	\)
\end{lemma}

We recall a result from~\cite{em79} which gives us an algorithm for computing
the value $\Regret{G}{\StrAllE,\StrPosA}$ in polynomial space. In~\cite{em79}
the authors show that the value of a mean-payoff game $G$ is equivalent to the
value of a finite \emph{cycle forming game} $\Gamma_G$ played on $G$. The game
is identical to the mean-payoff game except that it is finite. The game is
stopped as soon as a cycle is formed and the value of the game is given by the
mean-payoff value of the cycle.

\begin{proposition}[Finite Mean-Payoff Game~\cite{em79}]
	The value of a mean-payoff game $G$ is equal to the value of the finite
	cycle forming game $\Gamma_G$ played on the same weighted arena.
\end{proposition}

As $\linffun$ and
$\lsupfun$ games are also equivalent to their finite cycle forming game
(see~\cite{ar14}) it follows that one can use an Alternating Turing Machine to
compute the value of a game and that said machine will stop in time bounded by
the length of the largest simple cycle in the arena.  We note the length of the
longest simple path in $\widehat{G}$ is bounded by $|V|(|E| + 1)$. Hence, we can
compute the winner of $\widehat{G}$ in alternating polynomial time. Since $\AP =
\PSPACE$, this concludes the proof of
Lemma~\ref{lem:pspace-memless-regret}.

\subsection{Lower bounds}
We give a reduction from the \textsc{QSAT Problem} to the problem of determining
whether, given $r \in \mathbb{Q}$, $\Regret{G}{\StrAllE,\StrPosA} \lhd r$ for
the payoff functions $\linffun$, $\underline{\mpfun}$, and $\overline{\mpfun}$
(for $\lhd \in \{<,\le\}$).  Then we provide a reduction from the
complement of  the \textsc{$2$-disjoint-paths Problem} for $\lsupfun$,
$\supfun$, and $\inffun$.

The crux of the reduction from QSAT is a gadget for each clause of the
QSAT formula. Visiting this gadget allows \eve to gain information
about the highest payoff obtainable in the gadget, each entry point
corresponds to a literal from the clause, and the literal is visited
when it is made true by the valuation of variables chosen by \eve and
\adam in the reduction described below. Figure~\ref{fig:clause-gadget}
depicts an instance of the gadget for a particular clause.
Let us focus on the mean-payoff function.  Note that staying in the
inner $6$-vertex triangle would yield a mean-payoff value of $4$.
However, in order to do so, \adam needs to cooperate with \eve at all
three corner vertices. Also note that if he does cooperate
in at least one of these vertices then
\eve can secure a payoff value of at least $\frac{11}{3}$.

\begin{lemma}
	\label{lem:pspace-hardness}
	Given $r \in \mathbb{Q}$ and a weighted arena $G$ with payoff
	function $\linffun$, $\underline{\mpfun}$, or $\overline{\mpfun}$,
	determining whether $\Regret{G}{\StrAllE,\StrPosA} \lhd r$, for
	$\lhd \in \{<,\le\}$, is \PSPACE-hard.
\end{lemma}

The \textsc{QSAT Problem} asks whether a given fully quantified boolean
formula (QBF) is satisfiable. The problem is known to be
\PSPACE-complete~\cite{gj79}. It is known the result holds even if the
formula is assumed to be in conjunctive normal form with three literals
per clause (also known as $3$-CNF). Therefore, w.l.o.g., we consider an instance
of the \textsc{QSAT Problem} to be given in the following form:
\[
	\exists x_0 \forall x_1 \exists x_2 \ldots \Phi(x_0, x_1,
	\ldots, x_n)
\]
where $\Phi$ is in $3$-CNF. Also w.l.o.g., we assume
that every non-trivially true clause has at least one existentially
quantified variable (as otherwise the answer to the problem is trivial).

It is common to consider a QBF as a game between an existential and a
universal player. The existential player chooses a truth value for
existentially quantified variable $x_i$ and the universal player
responds with a truth value for $x_{i+1}$. After $n$ turns the truth
value of $\Phi$ determines the winner: the existential player wins if
$\Phi$ is true and the universal player wins otherwise. The game we
shall construct mimics the choices of the existential and universal
player and makes sure that the regret of the game is small if and only
if $\Phi$ is true.

Let us first consider the strict regret threshold problem. We will construct a
weighted arena $G$ in which \eve wins in the strict regret threshold problem
for threshold $2$ if and only if $\Phi$ is true.
\begin{lemma}
	Given a weighted arena $G$ with payoff function $\linffun$,
	$\underline{\mpfun}$, or $\overline{\mpfun}$, determining whether
	$\Regret{G}{\StrAllE,\StrPosA} < 2$ is \PSPACE-hard.
\end{lemma}

\begin{figure}
\begin{center}
\begin{tikzpicture}[clause/.style={isosceles triangle, shape border
	rotate=90,inner sep=3pt, draw,dotted},node distance=0.8cm]
\node[ve,initial](A){};
\node[ve, right=of A, yshift=-0.5cm](B){$x_0$};
\node[ve, right=of A, yshift=0.5cm](B'){$\overline{x_0}$};
\node[va, right=of B, yshift=0.5cm](C){};
\node[ve, right=of C, yshift=-0.5cm](D){$x_1$};
\node[ve, right=of C, yshift=0.5cm](D'){$\overline{x_1}$};
\node[right=of D, yshift=0.5cm](E){$\cdots$};
\node[ve,right=of E, yshift=-0.5cm](F){$x_n$};
\node[ve,right=of E, yshift=0.5cm](F'){$\overline{x_n}$};
\node[ve,right=of F, yshift=0.5cm](G){$\Phi$};

\node[below=0.5cm of B, clause,xshift=-0.8cm](X){$C_i$};
\node[below=0.5cm of B](Y){$\cdots$};
\node[below=0.5cm of B, clause,xshift= 0.8cm](Z){$C_j$};

\path
(A) edge (B)
(B) edge (C)
(A) edge (B')
(B') edge (C)
(C) edge (D)
(C) edge (D')
(D) edge (E)
(D') edge (E)
(E) edge (F)
(E) edge (F')
(F) edge (G)
(F') edge (G)
(G) edge[loop right] node[el,swap]{$2$} (G)
(B) edge[bend left] (X.apex)
(B) edge[bend left] (Z.apex)
(X.apex) edge[bend left] (B)
(Z.apex) edge[bend left] (B)
;

\end{tikzpicture}
\caption{Depiction of the reduction from QBF.}
\label{fig:spine-qbf}
\end{center}
\end{figure}

\begin{figure}
\begin{center}
\begin{tikzpicture}
\node[ve](A) at (2,5) {\small{$x_i$}};
\node(rA) at (4,6) {};
\node(lA) at (0,6) {};
\node[va](B) at (2,3) {};
\node[ve](C) at (3,2) {};
\node[va](D) at (4,1) {};
\node[ve](E) at (6,0) {\small{$\overline{x_j}$}};
\node(bE) at (7,-2) {};
\node(aE) at (7,2) {};
\node[ve](F) at (2,1) {};
\node[va](G) at (0,1) {};
\node[ve](H) at (1,2) {};
\node[ve](I) at (-2,0) {\small{$x_k$}};
\node(bI) at (-3,-2) {};
\node(aI) at (-3,2) {};

\path
(A) edge[dotted] node[el]{$0$} (rA)
(lA) edge[dotted] node[el]{$0$} (A)
(E) edge[dotted] node[el]{$0$} (bE)
(aE) edge[dotted] node[el]{$0$} (E)
(I) edge[dotted] node[el]{$0$} (aI)
(bI) edge[dotted] node[el]{$0$} (I)
(A) edge[bend left] node[el]{$3$} (B)
(B) edge[bend left] node[el]{$0$} (A)
(B) edge node[el]{$4$} (C)
(C) edge[bend right] node[el,swap]{$4$} (A)
(C) edge node[el]{$4$} (D)
(D) edge[bend left] node[el]{$0$} (E)
(E) edge[bend left,pos=0.7] node[el]{$3$} (D)
(D) edge node[el]{$4$} (F)
(F) edge[bend right] node[el,swap]{$4$} (E)
(F) edge node[el]{$4$} (G)
(G) edge node[el,swap]{$4$} (H)
(H) edge node[el]{$4$} (B)
(G) edge[bend left] node[el]{$0$} (I)
(I) edge[bend left] node[el,swap]{$3$} (G)
(H) edge[bend right,pos=0.2] node[el,swap]{$4$} (I);
\end{tikzpicture}
\caption{Clause gadget for the QBF reduction for clause $x_i \lor \lnot x_j
\lor x_k$.}
\label{fig:clause-gadget}
\end{center}
\end{figure}

\begin{proof}
	We first describe the value-choosing part of the game (see
	Figure~\ref{fig:spine-qbf}). $\VtcE$ contains vertices for every
	existentially quantified variable from the QBF and $\VtcA$ contains
	vertices for every universally quantified variable. At each of this
	vertices, there are two outgoing edges with weight $0$ corresponding to
	a choice of truth value for the variable. For the variable $x_i$ vertex,
	the \textbf{true} edge leads to a vertex from which \eve can choose to
	move to any of the clause gadgets corresponding to clauses where the
	literal $x_i$ occurs (see dotted incoming edge in
	Figure~\ref{fig:clause-gadget}) or to advance to $x_{i+1}$. The
	\textbf{false} edge construction is similar, while leading to the
	literal $\overline{x_i}$ rather than to $x_i$. From the vertices encoding
	the choice of truth value for $x_n$ \eve can either visit the clause
	gadgets for it or move to a ``final'' vertex $\Phi \in \VtcE$. This
	final vertex has a self-loop with weight $2$.

	To conclude the proof, we describe the strategy of \eve which ensures
	the desired property if the QBF is satisfiable and a strategy of \adam
	which ensures the property is falsified otherwise.
	
	Assume the QBF is true. It follows that there is a strategy of
	the existential player in the QBF game such that for any strategy of the
	universal player the QBF will be true after they both choose values for
	the variables. \eve now follows this strategy while visiting all clause
	gadgets corresponding to occurrences of chosen literals. At every gadget
	clause she visits she chooses to enter the gadget. If \adam
	now decides to take the weight $4$ edge, \eve can achieve a mean-payoff
	value of $\frac{11}{3}$ or a $\linffun$ value of $3$ by staying in the
	gadget. In this case the claim trivially holds. We therefore focus in
	the case where \adam chooses to take \eve back to the vertex from which
	she entered the gadget. She can now go to the next clause gadget and
	repeat. Thus, when the play reaches vertex $\Phi$, \eve must have
	visited every clause gadget and \adam has chosen to disallow a weight
	$4$ edge in every gadget. Now \eve can ensure a payoff value of $2$ by
	going to $\Phi$. As she has witnessed that in every clause gadget there
	is at least one vertex in which \adam is not helping her, alternative
	strategies might have ensured a mean-payoff of at most $\frac{11}{3}$
	and a $\linffun$ value of at most $3$. Thus, her regret is less than
	$2$.

	Conversely, if the universal player had a winning strategy (or, in other
	words, the QBF was not satisfiable) then the strategy of \adam consists
	in following this strategy in choosing values for the variables and
	taking \eve out of clause gadgets if she ever enters one. If the play
	arrives at $\Phi$ we have that there is at least one clause gadget that
	was not visited by the play. We note there is an alternative strategy of
	\eve which, by choosing a different valuation of some variable, reaches
	this clause gadget and with the help of \adam achieves value $4$. Hence,
	this strategy of \adam ensures regret of exactly $2$. If \eve avoids
	reaching $\Phi$ then she can ensure a value of at most $0$, which means
	an even greater regret for her.
\end{proof}

We observe that the above reduction can be readily parameterized. That is, we can
replace the $4$ value, the $3$ value and the $2$ value by arbitrary values $A$,
$B$, $C$ satisfying the following constraints:
\begin{itemize}
	\item $A > B > C$,
	\item $\frac{2A + B}{3} - C < r$ so that \eve wins if $\Phi$ is true,
	\item $A - C \ge r$ so that \adam wins if $\Phi$ is false, and
	\item $A - \frac{2A + B}{3} < r$ so that he never helps \eve in the
		clause gadgets.
\end{itemize}
Indeed, the valuation of $A$, $B$, $C$ we chose: $4$, $3$, $2$ with $r = 2$,
satisfies these inequalities exactly. It is not hard to see that if we find a
valuation for $r$, $A$, $B$, $C$ which meets the first restriction and the last
three having changed from strict to non-strict, and vice-versa, we can get a
reduction that works for the non-strict regret threshold problem. That is, find
values such that
\begin{itemize}
	\item $A > B > C$,
	\item $\frac{2A + B}{3} - C \le r$ so that \eve wins if $\Phi$ is true,
	\item $A - C > r$ so that \adam wins if $\Phi$ is false, and
	\item $A - \frac{2A + B}{3} \le r$ so that he never helps \eve in the
		clause gadgets.
\end{itemize}
For example, one could consider $A = 10$, $B = 7$, $C = 5$ and $r = 4$.

\begin{lemma}
	Given a weighted arena $G$ with payoff function $\linffun$,
	$\underline{\mpfun}$, or $\overline{\mpfun}$, determining whether
	$\Regret{G}{\StrAllE,\StrPosA} \le 4$ is \PSPACE-hard.
\end{lemma}

\begin{figure}
\begin{center}
\begin{tikzpicture}
	\node[va] (A) at (0,1.5) {$v$};
	\node[ve] (B) at (2,1.5) {};
	\node[va] (C) at (4,1.5) {$t_1$};
	\node[va] (D) at (2,0) {};
	\node[va] (E) at (4,0) {$s_2$};
	
	\path
	(A) edge (B)
	(B) edge (C)
	(B) edge (D)
	(D) edge[loop left] (D)
	(D) edge (E)
	;
\end{tikzpicture}
\caption{Regret gadget for $2$-disjoint-paths reduction.}
\label{fig:2dp-gadget}
\end{center}
\end{figure}

\begin{lemma}
	\label{lem:coNP-hardness}
	Given $r \in \mathbb{Q}$ and a weighted arena $G$ with payoff
	function $\inffun$, $\supfun$, or $\lsupfun$,
	determining whether
	$\Regret{G}{\StrAllE,\StrPosA} \lhd r$, for $\lhd \in \{<,\le\}$,
	is \coNP-hard.
\end{lemma}
\begin{proof}
	We provide a reduction from the complement of the
	\textsc{$2$-disjoint-paths Problem} on directed
	graphs~\cite{eilam-tzoreff98}. As the problem is known to be
	\NP-complete, the result follows. In other words, we sketch how to
	translate a given instance of the \textsc{$2$-disjoint-paths Problem}
	into a weighted arena in which \eve can ensure regret value strictly
	less than $1$ if and only if the answer to the
	\textsc{$2$-disjoint-paths Problem} is negative.

	Consider a directed graph $G$ and distinct vertex pairs $(s_1,t_1)$ and
	$(s_2,t_2)$. W.l.o.g. we assume that for all $i \in \{1,2\}$:
	\begin{inparaenum}[$(i)$]
		\item $t_i$ is reachable from $s_i$, and
		\item $t_i$ is a sink (i.e. has no outgoing edges).
	\end{inparaenum}
	in $G$. We now describe the changes we apply to $G$ in order to get the
	underlying graph structure of the weighted arena and then comment on the
	weight function. Let all vertices from $G$ be \adam vertices and $s_1$
	be the initial vertex. We replace all edges on $t_1$---edges of
	the form  $(v, t_1)$ incident, for some $v$---by a copy of the gadget
	shown in Figure~\ref{fig:2dp-gadget}. Next, we add self-loops on $t_1$
	and $t_2$ with weights $1$ and $2$, respectively. Finally, the weights
	of all remaining edges are $0$.

	We claim that, in this weighted arena, \eve can ensure regret
	strictly less than $1$---for payoff functions $\supfun$ and
	$\lsupfun$---if and only if in $G$ the vertex pairs $(s_1,t_1)$ and $(s_2,t_2)$
	cannot be joined by vertex-disjoint paths.  Indeed, we claim that the
	strategy that minimizes the regret of \eve is the strategy that, in
	states where she has a choice, tells her to go to $t_1$.
	
	First, let us prove that this strategy has regret strictly less than $1$
	if and only if no two disjoint paths in the graph exist between the
	pairs of states $(s_1,t_1)$ and $(s_2,t_2)$. Assume the latter is the
	case. Then if \adam chooses to always avoid $t_1$, then clearly the
	regret is $0$. If $t_1$ is eventually reached, then the choice of \eve
	secures a value of $1$ (for all payoff functions). Note that if she had
	chosen to go towards $s_2$ instead, as there are no two disjoint paths,
	we know that either the path constructed from $s_2$ by \adam never
	reaches $t_2$, and then the value of the path is $0$---and the regret
	is $0$ for \eve---or the path constructed from $s_2$ reaches $t_1$
	again since \adam is playing positionally---and, again, the regret is
	$0$ for \eve. Now assume that two disjoint paths between the
	source-target pairs exist. If \eve changed her strategy to go towards
	$s_2$ (instead of choosing $t_1$) then \adam has a strategy to reach
	$t_2$ and achieve a payoff of $2$. Thus, her regret would be equal to
	$1$.

	Second, we claim that any other strategy of \eve has a regret greater
	than or equal to $1$. Indeed, if \eve decides to go towards $s_2$
	(instead of choosing to go to $t_1$) then \adam can choose to loop on
	the state before $s_2$ and the payoff in this case is $0$. Hence, the
	regret of \eve is at least $1$.

	Note that minimal changes are required for the same construction to
	imply the result for $\inffun$. Further, the weight function and
	threshold $r$ can be accommodated so that \eve wins for the non-strict
	regret threshold. Hence, the general result follows.
\end{proof}

\subsection{Memory requirements for \eve} 
It follows from our algorithms for computing regret in this variant that Eve only
requires strategies with exponential memory. 
Examples where exponential memory is necessary can be easily constructed.

\begin{corollary}
	For all payoff functions $\supfun$, $\inffun$, $\lsupfun$, $\linffun$,
	$\underline{\mpfun}$ and $\overline{\mpfun}$, for all game graphs $G$,
	there exists $m$ which is  $2^{\mathcal{O}(|G|)}$ such that:
	\[
		\Regret{G}{\StrAllE,\StrPosA } = \Regret{G}{\StrE^m,\StrPosA}.
	\]
\end{corollary}

\section{Variant III: \adam plays word strategies}
\label{sec:word-strats}
For this variant, we provide tight upper and lower bounds for all the payoff
functions: the regret threshold problem is $\EXP$-complete for $\supfun$,
$\inffun$, $\lsupfun$, and $\linffun$, and undecidable for $\underline{\mpfun}$
and $\overline{\mpfun}$.  For the later case, the decidability can be recovered
when we fix a priori the size of the memory that \eve can use to play, the
decision problem is then $\NP$-complete.  Finally, we show that our notion of
regret minimization for word strategies generalizes the notion of {\em good for
games} introduced by Henzinger and Piterman in~\cite{hp06}, and we also
formalize the relation that exists with the notion of determinization by pruning
for weighted automata introduced by Aminof et al. in~\cite{akl10}. 

\paragraph{Additional definitions.}
We say that a strategy of \adam is a \emph{word strategy} if his strategy can be
expressed as a function $\tau : \mathbb{N} \to [\max\{\outdeg(v) \st v \in
V\}]$, where $[n] = \{i \st 1 \le i \le n\}$ and $\outdeg(v)$ is the
\emph{outdegree} of $v$ (i.e. the number of edges leaving $v$).  Intuitively, we
consider an order on the successors of each \adam vertex. On every turn, the
strategy $\tau$ of \adam will tell him to move to the $i$-th successor (or to a
sink state, if its outdegree is less than $i$) of the vertex according to the fixed order. We denote by
$\StrWordA$ the set of all such strategies for \adam. When considering word
strategies, it is more natural to see the arena as a (weighted) automaton.

A \emph{weighted automaton} is a tuple $\Gamma=(Q, q_I, A, \Delta, w)$ where $A$
is a finite alphabet, $Q$ is a finite set of states, $q_I$ is the initial state,
$\Delta \subseteq Q \times A \times Q$ is the transition relation, $w : \Delta
\rightarrow \mathbb{Q}$ assigns weights to transitions. A \emph{run} of $\Gamma$
on a word $a_0 a_1 \ldots \in A^\omega$ is a sequence $\rho = q_0 a_0 q_1 a_1
\ldots \in (Q\times A)^\omega$ such that $(q_i,a_i,q_{i+1}) \in \Delta$, for all
$i \ge 0$, and has \emph{value} $\Val(\rho)$ determined by the sequence of
weights of the transitions of the run and the payoff function. The value
$\Gamma$ assigns to a word is the supremum of the values of all its runs on the
word. We say the automaton is deterministic if $\Delta$ is functional.

A game in which \adam plays word strategies can be
reformulated as a game played on a weighted automaton $\Gamma=(Q, q_I, A,
\Delta, w)$ and strategies of \adam\ -- of the form $\tau : \mathbb{N} \to A$ --
determine a sequence of input symbols to which \eve has to react by choosing
$\Delta$-successor states starting from $q_I$. In this setting a strategy of
\eve which minimizes regret defines a run by resolving the non-determinism of
$\Delta$ in $\Gamma$, and ensures the difference of value given by the
constructed run is minimal w.r.t. the value of the best run on the word
spelled out by \adam. For instance, if all vertices in Fig.~\ref{fig:intro-ex}
are replaced by states, \eve can choose the successor of $v_1$ regardless of
what letter \adam plays and from $v_2$ and $v_3$ \adam chooses the successor by
choosing to play $a$ or $b$. Furthermore, his choice of letter tells \eve what
would have happened had the play been at the other state.

The following result summarizes the results of this section:

\begin{theorem}\label{thm:eloquentadversary}
	Deciding whether the regret value is less than a given threshold (strictly or
	non-strictly) playing against word strategies of \adam is \EXP-complete
	for $\inffun$, $\supfun$, $\linffun$, and $\lsupfun$; it is undecidable
	for $\underline{\mpfun}$ and $\overline{\mpfun}$.
\end{theorem}

\subsection{Upper bounds}
There is an $\EXP$ algorithm for solving the regret threshold problem for
$\inffun$, $\supfun$, $\linffun$, and $\lsupfun$. This algorithm is obtained by a
reduction to parity and Streett games. 

\begin{lemma}\label{lem:exp-memb}
	Given $r \in \mathbb{Q}$ and a weighted automaton $\Gamma$ with
	payoff function $\inffun$, $\supfun$, $\linffun$, or $\lsupfun$,
	determining whether \( \Regret{\Gamma}{\StrAllE,\StrWordA} \lhd r \),
	for $\lhd \in \{<,\le\}$, can be done in exponential time.
\end{lemma}

We show how to decide the strict regret threshold problem. However, the same
algorithm can be adapted for the non-strict version by changing strictness of
the inequalities used to define the parity/Streett accepting conditions. 

\begin{proof}
	We focus on the $\linffun$ and $\lsupfun$ payoff functions. The result
	for $\inffun$ and $\supfun$ follows from the translation to $\linffun$
	and $\lsupfun$
	games given in Sect.~\ref{sec:all-strats}. Our decision algorithm
	consists in first building a deterministic automaton for
	$\Gamma=(Q_1,q_I,A,\Delta_1,w_1)$ using the construction provided
	in~\cite{cdh10}. We denote by $D_{\Gamma}=(Q_2,s_I,A,\Delta_2,w_2)$ this
	deterministic automaton and we know that it is at most exponentially
	larger than $\Gamma$.  Next, we consider a simulation game played by
	\eve and \adam on the automata $\Gamma$ and $D_{\Gamma}$. The game is
	played for an infinite number of rounds and builds runs in the two
	automata, it starts with the two automata in their respective initial
	states $(q_I,s_I)$, and if the current states are $q_1$ and $q_2$, then
	the next round is played as follows:
	\begin{itemize}
		\item \adam chooses a letter $a \in A$, and the state of the
			deterministic automaton is updated accordingly, i.e.
			$q'_2=\Delta_2(q_2,a)$, then
		\item  \eve updates the state of the non-deterministic automaton
			to $q'_1$ by reading $a$ using one of the edges labelled
			with $a$ in the current state, i.e. she chooses $q'_1$
			such that $q'_1 \in \Delta_1(q_1,a)$. The new state of
			the game is $(q'_1,q'_2)$.  
	\end{itemize}
	\eve wins the simulation game if the minimal weight seen infinitely
	often in the run of the non-deterministic automaton is larger than or
	equal to the minimal weight seen infinitely often in the deterministic
	automaton minus $r$.  It should be clear that this happens exactly when
	\eve has a regret bounded by $r$ in the original regret game on the word
	which is spelled out by \adam.

	Let us focus on the $\liminf$ payoff function now. We will sketch how
	this game can be translated into a parity game. For completeness, we now
	provide a formal definition of the latter.
	A parity game is a pair $(G,\Omega)$ where $G$ is
	a non-weighted arena and $\Omega : V \to \mathbb{N}$ is a function that
	assigns a priority to each vertex. Plays, strategies, and other notions
	are defined as with games played on weighted arenas. A play in a parity
	game induces an infinite sequence of priorities. We say a play is
	winning for \eve if and only if the minimal priority seen infinitely
	often is odd. The \emph{parity index} of a parity game is the number of
	priorities labelling its vertices, that is $|\{\Omega(v) \st v \in V\}|$.

	To obtain
	the translation, we keep the structure of the game as above but we
	assign priorities to the edges of the games instead of weights. We do it
	in the following way. If $X=\{ x_1,x_2, \dots,x_n\}$ is the ordered set
	of weight values that appear in the automata (note that $|X|$ is bounded
	by the number of edges in the non-deterministic automaton), then we need
	the set of priorities $D=\{ 2, \dots, 2n+1\}$. We assign priorities to
	edges in the game as follows:
	\begin{itemize}
	      \item when \adam chooses a letter $a$ from $q_2$, then if the
	      	weight that labels the edge that leaves $q_2$ with
	      	letter $a$ in the deterministic automaton is equal to
	      	$x_i \in X$, then the priority is set to $2i+1$,
	      \item when \eve updates the non-deterministic automaton from
	      	$q_1$ with a edge labelled with weight $w$, then the
	      	color is set to $2i$ where $i$ is the index in $X$ such
	      	that $x_{i-1} \leq w+r < x_{i}$.
	\end{itemize}
	It should be clear then along a run, the minimal color seen infinitely
	often is odd if and only if the corresponding run is winning for \eve in
	the simulation game. So, now it remains to solve a parity game with
	exponentially many states and polynomially many priorities  w.r.t. the size
	of $\Gamma$. This can be done in exponential time with classical
	algorithms for parity games.

	\item \paragraph{$\lsupfun$ to Streett games.}
	Let us now focus on $\lsupfun$. In this case we will reduce our problem
	to that of determining the winner of a Streett game with state-space
	exponential w.r.t.  the original game but with number of Streett pairs
	polynomial (w.r.t. the original game).  Recall that a Streett game is a
	pair $(G,\calF)$ where $G$ is a game graph (with no weight function) and
	$\calF \subseteq \pow(V) \times \pow(V)$ is a set of Streett pairs. We
	say a play is winning for \eve if and only if for all pairs $(E,F) \in
	\calF$, if a vertex in $E$ is visited infinitely often then some vertex
	in $F$ is visited infinitely often as well.

	Consider a $\lsupfun$ automaton $\Gamma = (Q,q_I,A,\Delta,w)$. For $x_i
	\in \{w(d) \st d \in \Delta \}$ let us denote by $\calA^{\ge x_i}$ the
	B\"uchi automaton with B\"uchi transition set equivalent to all
	transitions with weight of at least $x_i$. We denote by $\calD^{\ge x_i} =
	(Q_i,q_{i,I},A,\delta_i,\Omega_i)$
	the deterministic parity automaton with the same language as $\calA^{\ge
	x_i}$.\footnote{Since $\delta_i$ is deterministic, we sometimes write
	$\delta_i(p,a)$ to denote the unique $q \in Q_i$ such that $(p,a,q) \in
	\delta_i$.} From~\cite{piterman07} we have that $\calD^{\ge x_i}$ has at
	most $2|Q|^{|Q|}|Q|!$ states and parity index $2|Q|$ (the number of
	priorities). Now, let $x_1 <
	x_2 < \dots < x_l$ be the weights appearing in transitions of $\Gamma$.
	We construct the (non-weighted) arena $G_\Gamma = (V, V_\exists,E,v_I)$
	and Streett pair set $\calF$ as follows
	\begin{itemize}
		\item $V = Q \times \prod_{i=1}^l Q_i \cup Q \times
			\prod_{i=1}^l Q_i \times A \cup Q \times
			\prod_{i=1}^l Q_i \times A \times Q$;
		\item $V_\exists = Q \times \prod_{i=1}^l Q_i \times A$;
		\item $v_I = (q_I,q_{1,I},\dots,q_{l,I})$;
		\item $E$ contains
			\begin{itemize}
				\item $\big((p,p_1,\dots,p_l)),
					(p,p_1,\dots,p_l,a)\big)$ for all $a \in
					A$,
				\item $\big((p,p_l,\dots,p_l,a),
					(p,p_1,\dots,p_l,a,q)\big)$ if $(p,a,q) 
					\in \Delta$,
				\item $\big((p,p_l,\dots,p_l,a,q),
					(q,q_1,\dots,q_l)\big)$
					if for all $1 \le i \le l$:
					$(p_i,a,q_i) \in \delta_i$;
			\end{itemize}
		\item For all $1 \le i \le l$ and all even $y$ such that
			$\range{\Omega_i} \ni y$,
			$\calF$ contains the pair $(E_i,F_i)$ where
			\begin{itemize}
				\item $E_{i,y} = \{(p,\ldots,p_i,\dots,p_l,a,q) \st
					\Omega_i(p_i,a,\delta(p_i,a)) =
					y\}$, and
				\item $F_{i,y} = \{(p,\ldots,p_j,\dots,p_l,a,q) \st
					(\Omega_i(p_i,a,\delta(p_i,a)) < y \land y
					\pmod{2} = 1) \lor w(p,a,q) \ge x_i - r \}$.
			\end{itemize}
	\end{itemize}
	It is not hard to show that in the resulting Streett game, a strategy
	$\sigma$ of \eve is winning against any strategy $\tau$ of \adam if and
	only if for every automaton $\calD^{\ge x_i}$ which accepts the word
	induced by $\tau$ then the run of $\Gamma$ induced by $\sigma$ has
	payoff of at least $x_i - r$, if and only if \eve has a winning strategy
	in $\Gamma$ to ensure regret is less than $r$.

	Note that the number of Streett pairs in $G_\Gamma$ is polynomial w.r.t.
	the size of $\Gamma$, i.e.
	\begin{align*}
		|\calF| &\le \sum_{i=0}^l |\range{\Omega_i}|\\
		&\le l \cdot 2|Q|\\
		&\le |Q|^2 \cdot 2|Q| = 2|Q|^3.
	\end{align*}
	From~\cite{pp06} we have that Streett games can be solved in time
	$\calO(mn^{k+1}kk!)$ where $n$ is the number of states, $m$ the number
	of transitions and $k$ the number of pairs in $\calF$. Thus, in this
	case we have that $G_\Gamma$ can be solved in 
	\[
		\calO\big((2|Q|^{|Q|}|Q|!)^{3 + 2|Q|^3}\cdot 2|Q|^3 \cdot
		(2|Q|^3)!\big).
	\]
	which is still exponential time w.r.t. the size of $\Gamma$.
\end{proof}

\subsection{Exponential lower bounds}
We first establish $\EXP$-hardness for the payoff functions $\inffun$,
$\supfun$, $\linffun$, and $\lsupfun$ by giving a reduction from countdown
games~\cite{jsl08}.  That is, we show
that given a countdown game, we can construct a game where \eve ensures regret
less than $2$ if and only if Counter wins in the original countdown game.
\begin{lemma}\label{lem:exp-hardness}
	Given $r \in \mathbb{Q}$ and a weighted automaton $\Gamma$ with
	payoff function $\inffun$, $\supfun$, $\linffun$, or $\lsupfun$,
	determining whether $ \Regret{\Gamma}{\StrAllE,\StrWordA} \lhd r$,
	for $\lhd \in \{<,\le\}$, is \EXP-hard.
\end{lemma}

\begin{figure}
\begin{center}
\begin{tikzpicture}
\node[state] (sink0) at (0,2) {{$\bot_0$}};
\node[state] (linter) at (2,2) {};
\node[state,initial above] (vi) at (4,2) {};
\node[state] (rinter) at (6,2) {};
\node[state] (sink2) at (8,2) {{$\bot_2$}};
\node (lbinter) at (2,0) {};
\node (rbinter) at (6,0) {};

\path
(sink0) edge[loop above] node[el,swap] {$A,0$} (sink0)
(sink2) edge[loop above] node[el,swap] {$A,2$} (sink2)
(vi) edge node[el,swap] {$A,0$} (linter)
(linter) edge node[el,swap] {$bail, 0$} (sink0)
(vi) edge node[el] {$A,0$} (rinter)
(rinter) edge node[el] {$bail, 0$} (sink2)
(linter) edge node[el,swap] {$A\setminus \{bail\},0$} (lbinter)
(rinter) edge node[el] {$A\setminus \{bail\},0$} (rbinter)
;
\end{tikzpicture}
\caption{Initial gadget used in reduction from countdown games.}
\label{fig:initial-gadget}
\end{center}
\end{figure}

\begin{figure}
\begin{center}
\begin{tikzpicture}
\node at (0,2) {\dots};
\node[state] (xi) at (4,0) {{$x_i$}};
\node[state] (notxi) at (4,4) {{$\overline{x_i}$}};
\node[state] (carryi) at (3,2) {};
\node[state] (sink2) at (1,4) {{$\bot_2$}};
\node at (6,2) {\dots};
\node[state] (xn) at (7,0) {{$\overline{x_n}$}};
\node[state] (notxn) at (7,4) {{$x_n$}};
\node[state] (sink22) at (9,4) {{$\bot_2$}};

\path
(xi) edge[out=190,in=-110] node[el]{$c_{i+1},0$} (sink2)
(xi) edge[bend left] node[swap,el, align=left]{$b_i,0$\\$c_i,0$} (carryi)
(carryi) edge node[el,pos=0.1] {$A \setminus \{c_{i+1}\},0$} (sink2)
(carryi) edge[bend left] node[el,swap] {$c_{i+1},0$} (notxi)
(notxi) edge[bend right] node[el] {$c_{i+1},0$} (sink2)
(notxi) edge[bend left] node[el, align=left]{$b_i,0$\\$c_i,0$} (xi)
(xn) edge node[swap,el,align=left]{$b_n,0$\\$c_n,0$} (notxn)
(notxn) edge node[el] {$A,0$} (sink22)
(sink2) edge[loop above] node[el,swap] {$A,2$} (sink2)
(sink22) edge[loop above] node[el,swap] {$A,2$} (sink22)
;
\end{tikzpicture}
\caption{Counter gadget.}
\label{fig:counter-gadget}
\end{center}
\end{figure}

\begin{figure}
\begin{center}
\begin{tikzpicture}
\node[state] (v46) at (4,6) {};
\node[state] (v26) at (2,6) {};
\node[state] (v04) at (0,4.5) {};
\node[state] (v02) at (0,2) {};
\node[state] (v42) at (4,2) {};
\node[state] (v62) at (6,2) {};
\node[state] (v40) at (4,0) {};

\path
(v46) edge node[swap,el] {$b_0,0$} (v26)
(v46) edge node[el] {$b_0,0$} (v42)
(v26) edge node[swap,el] {$c_1,0$} (v04)
(v26) edge node[el] {$c_1,0$} (v42)
(v04) edge node[swap,el] {$c_2,0$} (v02)
(v04) edge node[el] {$c_2,0$} (v42)
(v02) edge node[el] {$c_3,0$} (v42)
(v02) edge[dotted,bend right] node[el,swap] {$c_3,0$} (v42)
(v42) edge node[el,swap] {$b_4,0$} (v40)
(v42) edge node[el] {$b_4,0$} (v62)
(v62) edge[dotted,bend left] node[el] {$c_5,0$} (v40)
(v62) edge[dotted,bend left] node[el,swap] {$c_5,0$} (v40)
;
\end{tikzpicture}
\caption{Adder gadget: depicted $+9$.}
\label{fig:adder-gadget}
\end{center}
\end{figure}

Let us first formalize what a countdown game is. A countdown game $\calC$
consists of a weighted graph $(S,T)$, where $S$ is the set of states and $T
\subseteq S \times (\mathbb{N} \setminus \{0\}) \times S$ is the transition
relation, and a target value $N \in \mathbb{N}$. If $t = (s,d,s')
\in T$ then we say that the duration of the transition $t$ is $d$. A
configuration of a countdown game is a pair $(s,c)$, where $s \in S$ is a state
and $c \in \mathbb{N}$. A move of a countdown game from a configuration $(s,c)$
consists in player Counter choosing a duration $d$ such that $(s,d,s') \in T$
for some $s' \in S$ followed by player Spoiler choosing $s''$ such that
$(s,d,s'') \in T$, the new configuration is then $(s'', c + d)$. Counter wins if
the game reaches a configuration of the form $(s, N)$ and Spoiler wins if the
game reaches a configuration $(s,c)$ such that $c < N$ and for all $t = (s,d,
\cdot) \in T$ we have that $c + d > N$.

Deciding the winner in a countdown game $\calC$ from a configuration $(s,0)$ --
where $N$ and all durations in $\calC$ are given in binary -- is \EXP-complete.

\begin{proof}[Proof of Lemma~\ref{lem:exp-hardness}]
	Let us fix a countdown game $\calC = ((S,T),N)$ and let
	$n = \lfloor \log_2 N \rfloor + 2$.

	\paragraph{Simplifying assumptions.}
	Clearly, if Spoiler has a winning strategy and the game continues beyond
	his winning the game, then eventually a configuration $(s,c)$, such that
	$c \ge 2^n$, is reached. Thus, we can assume w.l.o.g.  that plays in
	$\calC$ which visit a configuration $(s,N)$ are winning for Counter and
	plays which don't visit a configuration $(s,N)$ but eventually get to a
	configuration $(s',c)$ such that $c \ge 2^n$ are winning for Spoiler.

	Additionally, we can also assume that $T$ in $\calC$ is total. That is
	to say, for all $s \in S$ there is some duration $d$ such that $(s,d,s')
	\in T$ for some $s' \in S$. If this were not the case then for every $s$
	with no outgoing transitions we could add a transition
	$(s,N+1,s_\bot)$ where $s_\bot$ is a newly added state. It is easy
	to see that either player has a winning strategy in this new game if and
	only if he has a winning strategy in the original game.

	\paragraph{Reduction.}
	We will now construct a weighted arena $\Gamma$ with $W = 2$ such that,
	in a regret game with payoff function $\supfun$ played on $\Gamma$, \eve
	can ensure regret value strictly less than $2$ if and only if Counter
	has a winning strategy in $\calC$.

	As all weights are $0$ in the arena we build, with the exception of
	self-loops on sinks, the result holds for $\supfun$, $\lsupfun$ and
	$\inffun$. We describe the changes required for the $\inf$ result at the
	end.

	\paragraph{Implementation.}
	The alphabet of the weighted arena $\Gamma = (Q,q_I,A,\Delta,w)$ is $A =
	\{b_i \st 0 \le i \le n\} \cup \{c_i \st 0 < i \le n\} \cup
	\{bail,choose\} \cup S$. We now describe the structure of $\Gamma$ (i.e.
	$Q$, $\Delta$ and $w$). 
	
	\textbf{Initial gadget.} Figure~\ref{fig:initial-gadget} depicts the
	initial state of the arena.  Here, \eve has the choice of playing left
	or right. If she plays to the left then \adam can play $bail$ and force
	her to $\bot_0$ while the alternative play resulting from her having
	chosen to go right goes to $\bot_2$. Hence, playing left already gives
	\adam a winning strategy to ensure regret $2$, so she plays to the
	right. If \adam now plays $bail$ then \eve can go to $\bot_2$ and as $W
	=2$ this implies the regret will be $0$. Therefore, \adam plays anything
	but $bail$.

	\textbf{Counter gadget.} Figure~\ref{fig:counter-gadget} shows the left
	sub-arena. All states from $\{\overline{x_i} \st 0 \le i \le n\}$ have
	incoming transitions from the left part of the initial gadget with
	symbol $A \setminus \{bail\}$ and weight $0$. Let $y_0 \ldots y_n \in
	\mathbb{B}$ be the (little-endian) binary representation of $N$, then
	for all $x_i$ such that $y_i = 1$ there is a transition from $x_i$ to
	$\bot_0$ with weight $0$ and symbol $bail$. Similarly, for all
	$\overline{x_i}$ such that $y_i = 0$ there is a transition from
	$\overline{x_i}$ to $\bot_0$ with weight $0$ and symbol $bail$. All the
	remaining transitions not shown in the figure cycle on the same state,
	e.g. $x_i$ goes to $x_i$ with symbol $choose$ and weight $0$. 

	The sub-arena we have just described corresponds to a counter gadget
	(little-endian encoding) which keeps track of the sum of the durations
	``spelled'' by \adam. At any point in time, the states of this sub-arena
	in which \eve believes alternative plays are now will represent the
	binary encoding of the current sum of durations.  Indeed, the initial
	gadget makes sure \eve plays into the right sub-arena and therefore she
	knows there are alternative play prefixes that could be at any of the
	$\overline{x_i}$ states. This corresponds to the $0$ value of the
	initial configuration.

	\textbf{Adder gadget.} Let us now focus on the right sub-arena in which
	\eve finds herself at the moment. The right transition with symbol
	$A \setminus \{bail\}$ from the initial gadget goes to state $s$ -- the
	initial state from $\calC$. It is easy to see how we can simulate
	Counter's choice of duration and Spoiler's choice of successor. From $s$
	there are transitions to every $(s,c)$, such that $(s,c,s') \in T$ for
	some $s' \in S$ in $\calC$, with symbol $choose$ and weight $0$.
	Transitions with all other symbols and weight $0$ going to $\bot_1$ -- a
	sink with a $1$-weight cycle with every symbol -- from $s$ ensure \adam
	plays $choose$, lest since $W = 2$ the regret of the game will be at
	most $1$ and \eve wins.

	Figure~\ref{fig:adder-gadget} shows how \eve forces \adam to ``spell''
	the duration $c$ of a transition of $\calC$ from $(s,c)$. For
	concreteness, assume that \eve has chosen duration $9$. The top source
	in Figure~\ref{fig:adder-gadget} is therefore the state $(s,9)$. Again,
	transitions with all the symbols not depicted go to $\bot_1$ with weight
	$0$ are added for all states except for the bottom sink. Hence, \adam
	will play $b_0$ and \eve has the choice of going straight down or moving
	to a state where \adam is forced to play $c_1$.  Recall from the
	description of the counter gadget that the belief of \eve encodes the
	binary representation of the current sum of delays. If she believes a
	play is in $x_1$ (and therefore none in $\overline{x_1}$) then after
	\adam plays $b_0$ it is important for her to make him play $c_1$ or this
	alternative play will end up in $\bot_2$. It will be clear from the
	construction that \adam always has a strategy to keep the play in the
	right sub-arena without reaching $\bot_1$ and therefore if any
	alternative play from the left sub-arena is able to reach $\bot_2$ then
	\adam wins (i.e. can ensure regret $2$). Thus, \eve decides to force
	\adam to play $c_1$. As the duration was $9$ this gadget now forces
	\adam to play $b_4$ and again presents the choice of forcing \adam to
	play $c_5$ to \eve. Clearly this can be generalized for any duration.
	This gadget in fact simulates a \emph{cascade configuration} of $n$
	$1$-bit adders.

	Finally, from the bottom sink in the adder gadget, we have transitions
	with symbols from $S$ with weight $0$ to the corresponding states (thus
	simulating Spoiler's choice of successor state). Additionally, with any
	symbol from $S$ and with weight $0$ \eve can also choose to go to a
	state $q_{bail}$ where \adam is forced to play $bail$ and \eve is
	forced into $\bot_0$.
	
	\paragraph{Argument.} Note that if the simulation of the counter has
	been faithful and the belief of \eve encodes the value $N$ then by
	playing $bail$, \adam forces all of the alternative plays in the
	left sub-arena into the $\bot_0$ sink. Hence, if Counter has a winning
	strategy and \eve faithfully simulates the $\calC$ she can force this
	outcome of all plays going to $\bot_0$. Note that from the right
	sub-arena we have that $\bot_2$ is not reachable and therefore the
	highest payoff achievable was $1$. Therefore, her regret is of at most
	$1$.

	Conversely, if both players faithfully simulate $\calC$ and the
	configuration $N$ is never reached, i.e. Spoiler had a winning strategy
	in $\calC$ then eventually some alternative play in the left sub-arena
	will reach $x_n$ and from there it will go to $\bot_2$. Again, the
	construction makes sure that \adam always has a strategy to keep the
	play in the right sub-arena from reaching $\bot_1$ and therefore this
	outcome yields a regret of $2$ for \eve.

	\paragraph{Changes for $\inffun$.}
	For the same reduction to work for the $\inffun$ payoff function we add
	an additional symbol $kick$ to the alphabet of $\Gamma$. We also add
	deterministic transitions with $kick$, from all states which are not
	sinks $\bot_x$ for some $x$, to $\bot_0$. Finally, all non-loop transitions
	in the initial gadget are now given a weight of $2$; the ones in the
	counter gadget are given a weight of $2$ as well; the ones in the adder
	gadget (i.e. right sub-arena) are given a weight of $1$.

	We observe that if Counter has a winning strategy in the original game
	$\calC$ then \eve still has a winning strategy in $\Gamma$. The
	additional symbol $kick$ allows \adam to force \eve into a $0$-loop but
	also ensures that all alternative plays also go to $\bot_0$, thus
	playing $kick$ is not beneficial to \adam unless an alternative play is
	already at $\bot_2$. Conversely, if Spoiler has a winning strategy in
	$\calC$ then \adam has a strategy to allow an alternative play to reach
	$\bot_2$ while \eve remains in the adder gadget. He can then play $kick$
	to ensure the payoff of \eve is $0$ and achieve a maximal regret of $2$.

	Once again, we observe that the above reduction can be readily
	parameterized.  That is, we can replace the $2$ value, the $1$ value and
	the $0$ value from the $\bot_2,\bot_1,\bot_0$ sink loops by arbitrary
	values $A$, $B$, $C$ satisfying the following constraints:
	\begin{itemize}
		\item $A > B > C$,
		\item $A - C \ge r$ so that \eve loses by going left in the initial
			gadget,
		\item $A - B < r$ so that she does not lose by faithfully
			simulating the adder if she has a winning strategy from
			the countdown game, or in other words: if \adam cheats
			then $A-B$ is low enough to punish him,
		\item $B - C < r$ so that she does not regret having faithfully
			simulated addition, that is, if she plays her winning
			strategy from the countdown game then she does not
			consider $B-C$ too high and regret it.
	\end{itemize}
	Changing the strictness of the last three constraints and finding a
	suitable valuation for $r$ and $A,B,C$ suffices for the reduction to
	work for the non-strict regret threshold problem. Such a valuation is
	given by $A=2$, $B=1$, $C=0$ with $r=1$.
\end{proof}

\subsection{Undecidability for mean payoff}

In what follows we will prove the following lemma.

\begin{lemma}\label{lem:undec-mpg}
	Let $r \in \mathbb{Q}$ with $r > 3$, $\lhd
	\in \{<,\le\}$, and the payoff function be $\underline{\mpfun}$ or $\overline{\mpfun}$. Given a weighted automaton $\Gamma$, determining
	whether $\Regret{\Gamma}{\StrAllE,\StrWordA} \mathrel{\lhd} r$ is
	undecidable even if Eve is only allowed to play finite memory
	strategies.
\end{lemma}

We give a
reduction from the \emph{threshold problem for mean-payoff games with
partial observation}. This problem was shown to be undecidable
in~\cite{ddgrt10,hpr14}. First, as a warm-up, we give a simpler reduction from the \emph{blind} sub-class of mean-payoff games with partial observation to the strict sub-cases of the lemma. Finally, we present a full proof of the lemma.

\subsubsection{On games with partial observation}
A partially observable mean-payoff game (POMPG, for short) $\mathcal{G}$ is a tuple
$(Q,q_I,A,\Delta,w,\mathit{Obs})$ where $(Q,q_I,A,\Delta,w)$ is a weighted automaton
such that $\Delta$ is total --- that is, for all $(p,a) \in Q \times A$ there exists $q \in Q$ such that $(p,a,q) \in \Delta$ --- and $\mathit{Obs} \subseteq \pow(Q)$ is a partition of $Q$ into \emph{observations}.
For a finite or infinite run $\rho = q_0 a_0 q_1 a_1 \dots$, we write $\mathit{obs}(\rho)$ to denote the unique observation-letter sequence $o_0 a_0 o_1 a_1 \dots$ such that $q_i \in o_i$ for all $i \geq 0$. In particular,  $\mathit{obs}(q)$ denotes the unique observation $o \in \mathit{Obs}$ such that $q \in o$.
An \emph{observation-based} strategy of \emph{Player $1$} is of the form $\sigma : (\mathit{Obs} \cdot
A)^* \cdot \mathit{Obs} \to A$. A run $\rho = q_0 a_0 q_1 a_1 \dots$ is \emph{consistent} with $\sigma$ if $a_i = \sigma(\mathit{obs}(q_0 a_0 \dots q_i))$ for all $i \geq 0$. POMPGs with a single trivial observation, i.e. with $\mathit{Obs} = Q$, are said to be \emph{blind}.

One can think of every round of the game proceeding as follows: \emph{Player $1$} chooses a letter $a$. Then, \emph{Player $2$} selects an $a$-successor $q$ of the current state $p$ (i.e. $(p,a,q) \in \Delta$) and reveals its observation $\mathit{obs}(q)$ to Player $1$.
In POMPGs, Player $1$ intuitively wants to {\em minimize} the mean payoff of all runs consistent with his strategy. In contrast, Player $2$ --- the agent who resolves local non-determinism after Player $1$ chooses letters --- wants to {\em maximize} this value. The latter gives rise to the following decision problem.

\paragraph{The threshold problem.}
In the (non-)strict threshold problem, we are given a POMPG $\mathcal{G}$ and a threshold
$t \in \mathbb{Q}$. We are then asked whether there exists an observation-based strategy $\sigma$ of Player $1$
such that for all runs $\rho$ consistent with $\sigma$ we have $\Val(\rho) < t$ (respectively, $\leq$). Both problems are
known to be undecidable (regardless of whether the payoff function is $\underline{\mpfun}$ or $\overline{\mpfun}$) in general~\cite{ddgrt10,hpr14}.

For a POMPG $\mathcal{G}$, we say an observation-based strategy $\sigma$ of Player $1$ is $t$-winning if it is such that all runs $\rho$ consistent with it satisfy $\Val(\rho) \leq t$.

\begin{figure}
\centering
\begin{tikzpicture}[yscale=0.7,initial text={},every state/.style={minimum size=0.5cm,inner sep=0}]
\node[state] (sink0) at (0,2) {{$\bot_0$}};
\node[state] (linter) at (2,2) {};
\node[state,initial above] (vi) at (4,2) {};
\node[state] (rinter) at (6,2) {};
\node[state] (sink2) at (8,2) {};
\node (lbinter) at (2,0) {Copy of $\mathcal{G}$};
\node[align=center] (rbinter) at (6,0) {Gadget that allows Eve\\
to force Adam's hand};

\path[auto,->]
(sink0) edge[loop above] node[swap] {$B,0$} (sink0)
(sink2) edge[loop above] node[swap] {$B,r+1$} (sink2)
(vi) edge node[swap] {$B,0$} (linter)
(linter) edge node[swap] {$\mathit{bail}, 0$} (sink0)
(vi) edge node {$B,0$} (rinter)
(rinter) edge node {$\mathit{bail}, 0$} (sink2)
(linter) edge node[swap] {$B\setminus \{\mathit{bail}\},0$} (lbinter)
(rinter) edge node {$B\setminus \{\mathit{bail}\},0$} (rbinter)
;
\end{tikzpicture}
\caption{Initial gadget used in the reduction from partial-observation games}
\label{fig:initial-gadget-pompg}
\end{figure}

\paragraph{Towards regret.}
To prove Lemma 14, we show how to translate any instance of the threshold problem with a POMPG $\mathcal{G} = (Q,q_I,A,\Delta,w,\mathit{Obs})$ and threshold $t \in \mathbb{Q}$ into a query
regarding the regret of Eve in a game played on a weighted automaton $\Gamma$.
Intuitively, in the game played by Eve and Adam on $\Gamma$, Eve will be able to ``force'' Adam to play according to some observation-based strategy of Player $1$ in $\mathcal{G}$. Otherwise, she will be able to reach a low-regret sink. A modified version of the weighted automaton $(Q,q_I,A,\Delta,w)$ will be a sub-automaton of $\Gamma$. Our construction is such that for Eve to ensure a regret of at most $t$, in the game played on $\Gamma$, it will be sufficient and necessary for her to make Adam play according to a $t$-winning strategy of Player $1$ in $\mathcal{G}$.

Below, we first consider a reduction from blind games. This is sufficient to prove undecidability of the strict cases (i.e. $\lhd$ is $<$). The reduction is technically simpler and will allow the reader to build intuition before moving to the proof of the non-strict cases (where $\lhd$ is $\leq$) which require a reduction from general POMPGs.

\subsubsection{Blind games and the strict sub-cases}
In this subsection we will prove the following.
\begin{proposition}\label{pro:warmup}
	Let $r \in \mathbb{Q}$ with $r > 3$ and the payoff function be $\underline{\mpfun}$ or $\overline{\mpfun}$. Given a weighted automaton $\Gamma$, determining
	whether $\Regret{\Gamma}{\StrAllE,\StrWordA} < r$ is
	undecidable.
\end{proposition}

The strict threshold problem for blind POMPGs is shown to be undecidable in~\cite{ddgrt10}.
Note that in blind games, observation-based strategies of Player $1$ are just infinite words $Qa_0Qa_1 \dots$ and so solving blind games is equivalent to solve the universality problem for automata. This fact is used in~\cite{cdehr10} to deduce that the \emph{non-strict universality problem for mean-payoff automata} is undecidable. We formally define the problem below.

\paragraph{Non-strict universality.}
The non-strict universality problem asks, given $\mathcal{A}$ and a threshold $t \in \mathbb{Q}$, whether for all words $\alpha \in A^\omega$ there exists a run $\rho$ of $\mathcal{A}$ on $\alpha$ such that $\Val(\rho) \geq t$. If the answer to the problem is positive, we say $\mathcal{A}$ is $t$-universal.

\paragraph{Simplifying assumptions.}
Before we prove the proposition, we first argue
that the non-strict universality problem remains undecidable
even if the threshold is a fixed value (instead of being part of the input) and under assumptions
regarding the largest weights the input automata
may have. To formalize the latter, for a
weighted automaton $\mathcal{A} = (Q,q_I,A,\Delta,w)$ we write $W_{\mathcal{A}} \coloneqq \max_{d \in \Delta} |w(d)|$.
\begin{lemma}\label{lem:simple}
    Let $t \in \mathbb{Q}$ with $t > 3$. Given a weighted
    automaton $\mathcal{A}$ such that $W_{\mathcal{A}} < t + 1$, determining whether $\mathcal{A}$ is not $t$-universal is undecidable.
\end{lemma}
\begin{proof}
    Consider an instance of the non-strict universality problem. That is, we have as input a weighted automaton $\mathcal{A} = (Q,q_I,A,\Delta,w)$ and a threshold $t' \in \mathbb{Q}$. We will construct a weighted automaton $\mathcal{A}''$ from $\mathcal{A}$, by modifying its weights only, so that $W_{\mathcal{A}''} < t + 1$.
    Then, we will prove that $\mathcal{A}''$ is not $t'$-universal if and only if $\mathcal{A}$ is not $t$-universal.
    The result will thus follow from the fact that non-strict universality is undecidable~\cite{cdehr10}.

    First, we construct --- using logarithmic space only --- a weighted automaton $\mathcal{A}' = (Q,q_I,A,\Delta,w')$ with
    \(
    w':(p,a,q) \mapsto w(p,a,q) - t'.
    \)
    We write $\Val'(\cdot)$ for the value of runs in $\mathcal{A}'$, i.e. using $w'$ instead of $w$.
    It is easy to see that for all runs $\rho$ of $\mathcal{A}$ we have $\Val'(\rho) = \Val(\rho) - t'$ for all $\alpha \in A^\omega$. It follows that
    $\mathcal{A}'$ is not $0$-universal if and only if $\mathcal{A}$ is not $t'$-universal.
    
    Next --- again using logarithmic space only --- we construct a second weighted automaton $\mathcal{A}'' = (Q,q_I,A,\Delta,w'')$ with:
    \[
    w'':(p,a,q) \mapsto
    \frac{w'(p,a,q)}{W_{\mathcal{A}'} + 1} + t.
    \]
    Note that, since $W_{\mathcal{A}'} = \max_{d \in \Delta} |w'(d)|$ by definition, $W_{\mathcal{A'}} \geq 0$ and therefore $\nicefrac{w'(p,a,q)}{W_\mathcal{A}' + 1}$ is well-defined (that is, the denominator is nonzero) for all $(p,a,q) \in \Delta$. Furthermore, $\nicefrac{w'(p,a,q)}{W_\mathcal{A}' + 1} < 1$, again for all $(p,a,q) \in \Delta$, by definition of $W_{\mathcal{A}'}$.
    It follows that $W_{\mathcal{A}''} < t + 1$. Finally, we know that $\nicefrac{1}{W_{\mathcal{A}'} + 1}$ is well-defined and strictly positive. Hence, we have that:
    \[
    \Val''(\rho) = \frac{\Val'(\rho)}{W_{\mathcal{A}'} + 1} + t
    \]
    for all runs $\rho$ of $\mathcal{A}$. It follows that $\mathcal{A}''$ is not $t$-universal if and only if $\mathcal{A}'$ is not $0$-universal.
\end{proof}

\begin{figure}
    \centering
    \begin{tikzpicture}[yscale=0.7,initial text={},every state/.style={minimum size=0.5cm,inner sep=0}]
        \node[state,initial above](qa){$q_a$};
        \node[state,left=2cm of qa,label={left:$\bot_{\mathit{cheat}}$}](sink){};
        \node[state,initial above,right= of qa](qb){$q_b$};
        \node[state,initial above,right= of qb](qc){$q_c$};
        \node[right= of qc] (etc) {\dots};
        \path[auto,->]
        (qa) edge[loop below] node{$a, 0$} (qa)
        (sink) edge[loop below] node{$B, r-1$} (sink)
        (qa) edge[swap] node{$B \setminus \{a\}$} (sink)
        (qa) edge[out=45,in=135] node{$a, 0$} (qb)
        (qb) edge[loop below] node{$b, 0$} (qb)
        (qb) edge[bend left,swap] node{$b, 0$} (qa)
        (qb) edge[out=45,in=135] node{$b, 0$} (qc)
        (qa) edge[swap,out=-60,in=240,looseness=2] node{$a, 0$} (qc)
        (qc) edge[bend left,swap] node{$c, 0$} (qb)
        (qc) edge[loop below] node{$c, 0$} (qc)
        (qc) edge node{$c, 0$} (etc)
        (qa) edge[swap,out=-75,in=255,looseness=2.3] node{$a,0$} (etc)
        (qb) edge[swap,out=-60,in=240,looseness=2] node{$b,0$} (etc)
        ;
    \end{tikzpicture}
    \caption{Sub-automaton which allows Eve to force Adam to play letters of her choosing; From $q_a$, reading $a$, there are transitions to every $q \in \{q_b : b \in A\}$}
    \label{fig:force-adam}
\end{figure}

\begin{proof}[Proof of \autoref{pro:warmup}]
    Consider an instance of the problem from~\autoref{lem:simple} with $t = r$. That is, we are given a weighted automaton $\mathcal{A} = (Q,q_I,A,\Delta,w)$ such that $W_\mathcal{A} < r + 1$ and we are
    asked whether $\mathcal{A}$ is not $r$-universal. Recall that the answer to the latter question is positive if and only if there exists a word $\alpha \in A^\omega$ such that all runs $\rho$ of $\mathcal{A}$ on $\alpha$ are such that $\Val(\rho) < r$. 
    We will construct a weighted arena $\Gamma$ over the alphabet $B \coloneqq A \cup \{\mathit{bail}\}$ such that $\Regret{\Gamma}{\StrAllE,\StrWordA} < r$ if and only if $\mathcal{A}$ is not $r$-universal. 
    
    The initial part of $\Gamma$ is depicted in \autoref{fig:initial-gadget-pompg}. This \emph{initial gadget} leads to a left and a right sub-arena. On the left, the gadget leads to the initial state of $\mathcal{A}$. On the right, the gadget leads to a sub-arena we describe presently (see \autoref{fig:force-adam} for a more visual description). The right sub-arena has states $\{q_a : a \in A\}$ and has $0$-weighted transitions $(q_a,a,q_b)$ for all $a,b \in A$. All states $\{q_a : a \in A \}$ are initial, i.e. the initial gadget has transitions to all of them on all letters $b \in B \setminus \{\mathit{bail}\}$. It also has $0$-weighted transitions $(q_a,c,\bot_{\mathit{cheat}})$ for all $a \in A$ and $c \in B \setminus\{a\}$; and self-loops $(\bot_{\mathit{cheat}},b,\bot_{\mathit{cheat}})$ with 
    weight $r-1$ for all $b \in B$. 
    
    It is important to note that the $(q_a,c,\bot_{\mathit{cheat}})$ transitions make it so that whenever Eve chooses a state $q_a$ in the right sub-arena of $\Gamma$ then Adam must play $a$. Otherwise, Eve reaches the state $\bot_\mathit{cheat}$ where she will obtain a value of $r-1$
    and have a regret of at most $r-1$. In the formal proof presented below, we argue that this threat allows Eve to force Adam's choice of letters.
    
    \paragraph{($\Longrightarrow$)} We actually prove the contrapositive:
    Assume that $\mathcal{A}$ is indeed $r$-universal, that is, for all words $\beta$, we have that $\mathcal{A}(\beta) \geq r$. In the initial gadget, if Eve goes to the left, then Adam plays $\mathit{bail}$ and her regret is  (strictly) larger than $r$. If she goes right, the run proceeds into the right sub-arena. There, whenever Eve chooses state $q_a$, Adam must play $a$. This allows Eve to force Adam to play a word $\beta$ of her choice. It follows that the run she constructs --- on the infinite word $\beta$ played by Adam --- stays in the states $\{q_a : a \in A\}$ and has value $0$. However, we have assumed that $\mathcal{A}$ is $r$-universal. Hence, the regret of Eve is at least $r$ as there is a run, on the word played by Adam, in the left sub-arena with a value of at least $r$.
    
    \paragraph{($\Longleftarrow$)}
    Assume now that $\mathcal{A}$ is not $r$-universal. Then, there exists a word $\alpha = a_0 a_1 a_2 \dots \in A^\omega$ such that all runs $\rho$ of $\mathcal{A}$ on $\alpha$ are such that $\Val(\rho) < r$. In the initial gadget, Eve will go to the right sub-arena. Hence, if Adam plays $\mathit{bail}$, the value of the run of Eve is optimal and her regret is $0$. Otherwise, if Adam does not play $\mathit{bail}$, the run constructed by Eve proceeds to the right sub-arena. There, her strategy consists in playing to the states $q_{a_0} q_{a_1} q_{a_2} \dots$ corresponding to the letters of $\alpha$. Note that the run she constructs stays in the states $\{q_a : a \in A\}$ if and only if Adam plays the word $\alpha$. We consider both cases in turn below.
    \begin{description}
        \item[If Adam does play $\alpha$] then Eve constructs a run with value $0$. By construction, all alternative runs on $\alpha$ in the right sub-arena have a value of at most $r-1$. This is because $r-1$ is the maximal weight of any transition in this part of the arena.
        Additionally, recall we have also assumed that all runs $\rho$ in the left sub-arena (that is, all runs of $\mathcal{A}$ on $\alpha$) have a value strictly smaller than $r$. It follows that Eve has a regret value strictly smaller than $r$.
        \item[If Adam does not play $\alpha$] then the run Eve constructs reaches $\bot_{\mathit{cheat}}$ and therefore has a value of $r-1$.
        Note that, by construction, this is the maximal value a run in the right sub-arena can have. In the left sub-arena, the maximal value a run can have is $W_{\mathcal{A}}$. To conclude we observe that, because of our assumptions, we have:
        \[
            W_{\mathcal{A}} - (r-1) < r+1-(r-1) = 0 < 3 < r,
        \]
        which means that Eve has a regret value strictly smaller than $r$.\qedhere
    \end{description}
\end{proof}

\subsubsection{The full proof}
Before delving into the proof, we make explicit a few simplifying assumptions
regarding POMPGs.

\paragraph{Simplifying assumptions.}
We start by arguing the threshold problem remains undecidable for any fixed threshold and under assumptions about the largest value of a run. (That is, we generalize the assumptions from~\autoref{lem:simple} to the POMPG setting.)
Additionally, it will be convenient to ensure that runs of the game make the observation sequence explicit by
way of some additional letters. Importantly, for the the general non-blind setting, we will have to double all transition weights (i.e. multiply them by two). This change will make it harder to obtain a useful bound on the largest absolute weight (as we did in~\autoref{lem:simple}). Instead, we will bound the values of all runs.
\begin{lemma}\label{lem:pompg-simple}
    Let $t \in \mathbb{Q}$ with $t > 3$. Given a POMPG $\mathcal{G}$ such that:
    \begin{enumerate}
    \item $\Val(\rho) \leq t + 1$ for all runs $\rho$ in $\mathcal{G}$,
    \item $\mathit{Obs} \subseteq A$, and
    \item all runs that avoid a designated sink state $\bot$ are of the following form: \[
    q_0 a_0 \langle q_1 o_1 \rangle o_1 q_1 a_2 \langle q_2 o_2 \rangle o_2 q_2 a_3 \dots \in (Q \cdot A \cdot (Q \times \mathit{Obs}) \cdot \mathit{Obs})^\omega
    \]
    and such that $\mathit{obs}(\langle q_i, o_i \rangle) = \mathit{obs}(q_i) = o_i$ for all $i \geq 0$,
    \end{enumerate}
    determining whether Player $1$ has a $t$-winning observation-based strategy $\sigma$, i.e. whether all runs $\rho$ consistent with $\sigma$ are such that $\Val(\rho) \leq t$, is undecidable.
\end{lemma}

\begin{figure}
    \centering
    \begin{tikzpicture}[yscale=0.7,initial text={},every state/.style={minimum size=0.5cm,inner sep=0}]
        \node[state,initial above](p){$p$};
        \node[state,below right= of p](q2){$q_2$};
        \node[state,below left= of p](q1){$q_1$};
        \begin{pgfonlayer}{background}
        \node[fit=(p),draw=blue,rounded corners=2mm,fill=blue!20,label={left:$o_0$}](o0){};
        \node[fit=(q1)(q2),draw=green,rounded corners=2mm,fill=green!20,label={left:$o_1$}](o1){};
        \end{pgfonlayer}
        \path[auto,->]
        (p) edge[bend right] node{$a$} (q1)
        (p) edge[swap,bend left] node{$a$} (q2)
        ;
    \end{tikzpicture}
    \quad
    \begin{tikzpicture}[yscale=0.7,initial text={},every state/.style={minimum size=0.5cm,inner sep=0}]
        \node[state,initial above](p){$p$};
        \node[state,below right=2cm of p,label={left:$\langle q_2,o_1 \rangle$}](q2o){};
        \node[state,below left=2cm of p,label={right:$\langle q_1,o_1 \rangle$}](q1o){};
        \node[state,below= of q2o](q2){$q_2$};
        \node[state,below= of q1o](q1){$q_1$};
        \node[state,right=2cm of q2o](bot){$\bot$};
        \begin{pgfonlayer}{background}
        \node[fit=(p),draw=blue,rounded corners=2mm,fill=blue!20,label={left:$o_0$}](o0){};
        \node[fit=(q1)(q2)(q1o)(q2o),draw=green,rounded corners=2mm,fill=green!20,label={left:$o_1$}](o1){};
        \node[fit=(bot),draw=black,rounded corners=2mm,fill=gray!20,label={above:$\mathit{obs}(\bot)$}](obot){};
        \end{pgfonlayer}
        \path[auto,->]
        (p) edge[bend right] node{$a$} (q1o)
        (p) edge[swap,bend left] node{$a$} (q2o)
        (q1o) edge[pos=0.8] node{$o_1$} (q1)
        (q2o) edge[swap,pos=0.8] node{$o_1$} (q2)
        (q1o) edge[bend right,swap,pos=0.4] node{$A \setminus \{o_1\}$} (bot)
        (q2o) edge node{$A \setminus \{o_1\}$} (bot)
        ;
    \end{tikzpicture}
    \caption{POMPG transformation which ensures that Player $1$ must declare every new observation by way of the letters: on the left, a POMGP; on the right, the transformed POMPG}
    \label{fig:declare-obs}
\end{figure}

\begin{proof}
Consider a POMPG $\mathcal{G}$ with a threshold $t' \in \mathbb{Q}$. Following the constructions from the proof of~\autoref{lem:simple}, we can obtain $\mathcal{G}' = (Q,q_I,A,\Delta,w,\mathit{Obs})$ which satisfies the first condition and such that Player $1$ has a $t$-winning observation-based strategy in $\mathcal{G}'$ if and only if Player $1$ has a $t'$-winning observation-based strategy in $\mathcal{G}$.

From $\mathcal{G}'$ we construct --- using logarithmic space --- a new game POMPG $\mathcal{G}'' = (Q',q_I,A',\Delta',w',\mathit{Obs}')$ in which Player $1$ now has to ``declare'' the observation of every newly visited state. To achieve this, we split all transitions in two so that Player $1$ plays a letter as before and then declares the new observation. (The construction is depicted in \autoref{fig:declare-obs}.) Additionally, we make sure the new intermediate states are in the same observation. Finally, we scale the weights so that runs still have their original value. Formally, we define the following:
\begin{itemize}
    \item $Q' \coloneqq Q \cup \{\langle q, \mathit{obs}(q) \rangle : q \in Q\} \cup \{\bot\}$;
    \item $A' \coloneqq A \cup \mathit{Obs}$;
    \item $\Delta' \subseteq \{(\langle p, \mathit{obs}(p) \rangle,\mathit{obs}(p),p) : p \in Q \} \cup \{(p, a, \langle q, \mathit{obs}(q) \rangle) : (p,a,q) \in \Delta\}$;
    \item $w' : (p, a, \langle q, \mathit{obs}(q) \rangle) \mapsto 2 \cdot w(p,a,q)$ and $w' : (\langle p, \mathit{obs}(p) \rangle,\mathit{obs}(p),p) \mapsto 0$;
    \item $\mathit{Obs}' \coloneqq \{ o \cup \{\langle q, \mathit{obs}(q) \rangle : q \in o\} \mid o \in \mathit{Obs}\} \cup \{\{\bot\}\}$.
\end{itemize}
To make $\Delta'$ total, we add $(t+1)$-weighted self-loops $(\bot,a,\bot)$ for all $a \in A'$; $0$-weighted transitions $(p,a,\bot)$ for all $p \in Q$ and all $a \in \mathit{Obs}$; and $0$-weighted transitions $(\langle p, \mathit{obs}(p)\rangle,a, \bot)$ for all $p \in Q$ and all $a \neq \mathit{obs}(p)$.

Note that there is a bijective mapping from runs $\rho$ of $\mathcal{G}'$ to runs of $\mathcal{G}''$ that do not reach $\bot$:
\[
    f : q_0a_0q_1a_1 \dots \mapsto q_0 a_0 \langle q_1, \mathit{obs}(q_1)\rangle \mathit{obs}(q_1) q_1 a_1 \langle q_2, \mathit{obs}(q_2) \rangle \mathit{obs}(q_2) q_2 a_3 \dots
\]
such that $\Val'(\rho) = \Val''(f(\rho))$, where $\Val'(\cdot)$ and $\Val''(\cdot)$ denote the value functions of $\mathcal{G}'$ and $\mathcal{G}''$ respectively.
Furthermore, the bijection $f$ can be lifted to observation-letter sequences and thus also to observation-based strategies of Player $1$ in both POMPGs. Hence, the following are equivalent.
\begin{itemize}
    \item Player $1$ has a $t$-winning observation-based strategy in  $\mathcal{G}''$.
    \item Player $1$ has an observation-based strategy $\sigma$ in $\mathcal{G}''$ such that all runs $\rho''$ consistent with $\sigma$ avoid $\bot$ and have $\Val''(\rho'') \leq t$.
    \item Player $1$ has a $t$-winning observation-based strategy in $\mathcal{G}'$.
    \item Player $1$ has a $t'$-winning observation-based strategy in $\mathcal{G}$.
\end{itemize}
Finally, note that our construction preserves the first assumption --- i.e. $\Val(\rho) \leq t + 1$ for all runs $\rho$ --- since we only multiplied the weights by $2$ and transitions have been split into two steps. Additionally, runs that reach $\bot$ have a value of $t +1$.
\end{proof}

We are now ready to prove Lemma 14. Below, we focus on the non-strict regret threshold problem since we have earlier shown the strict version of it to be undecidable. However, we note that the argument works for both versions with minor adaptations.

\begin{figure}
    \centering
    \begin{tikzpicture}[yscale=0.7,initial text={},every state/.style={minimum size=0.5cm,inner sep=0}]
        \node[state,initial above](qa){$q_a$};
        \node[state,left=2cm of qa,label={left:$\bot_{\mathit{cheat}}$}](sink){};
        \node[state,below= of qa](qap){$q'_a$};
        \node[state,initial above,right= of qa](qb){$q_b$};
        \node[state,initial above,right= of qb](qc){$q_c$};
        \node[right= of qc] (etc) {\dots};
        \path[auto,->]
            (sink) edge[loop above] node{$B, r$} (sink)
            (qa) edge[swap] node{$B \setminus \{a\}$} (sink)
            (qa) edge[swap] node{$a,0$} (qap)
            (qap) edge[bend left] node{$B \setminus \mathit{Obs}, 0$} (sink)
            (qap) edge[bend right,pos=0.7] node{$\mathit{Obs}$} (qb)
            (qap) edge[bend right,pos=0.7] node{$\mathit{Obs}$} (qc)
            (qap) edge[bend right,pos=0.7] node{$\mathit{Obs}$} (etc)
        ;
    \end{tikzpicture}
    \caption{Sub-automaton which allows Eve to make Adam play according to an observation-based strategy of Player $1$ in a POMPG; For clarity, only transitions from $q_a$, $q'_a$, and $\bot_{\mathit{cheat}}$ are depicted}
    \label{fig:force-adam-pompg}
\end{figure}

\begin{proof}[Proof of Lemma 14]
    Consider an instance of the problem from~\autoref{lem:pompg-simple} with $t = r$. That is, we are given a POMPG $\mathcal{G} = (Q,q_I,A,\Delta,w,\mathit{Obs})$ such that all three assumptions hold and we are asked whether there is an observation-based strategy of Player $1$ such that for all runs $\rho$ consistent with it we have $\Val(\rho) \leq r$. We construct a weighted arena $\Gamma$ such that $\Regret{\Gamma}{\StrAllE,\StrWordA} \leq r$ if and only if Player $1$ has an $r$-winning observation-based strategy in $\mathcal{G}$. Finally, the remark concerning finite memory will follow from the undecidability results for POMPGs~\cite{ddgrt10,hpr14} and (the proof of) \autoref{lem:pompg-simple}.
    
    The alphabet of $\Gamma$ is $B \coloneqq A \cup \{\mathit{bail}\}$. The initial part of $\Gamma$ is depicted in \autoref{fig:initial-gadget-pompg}. This \emph{initial gadget} leads to a left and a right sub-arena. On the left, the initial gadget leads to the initial state of a modified version of $\mathcal{G}$ in which we change the weight of the self-loops on $\bot$ to $0$. On the right, the initial gadget leads to a sub-arena depicted in \autoref{fig:force-adam-pompg}. It has states $\{q_a, q'_a : a \in A \setminus \mathit{Obs}\}$; $0$-weighted transitions $(q_a,a,q'_a)$ for all $a \in A \setminus \mathit{Obs}$; and $0$-weighted transitions $(q'_a,o,q_b)$ for all $a,b \in A \setminus \mathit{Obs}$ and all $o \in \mathit{Obs}$. All states $\{q_a : a \in A \setminus \mathit{Obs} \}$ are initial. We also add $0$-weighted transitions $(q_a,c,\bot_{\mathit{cheat}})$ for all $a \in A \setminus \mathit{Obs}$ and $c \in B \setminus\{a\}$; and self-loops $(\bot_{\mathit{cheat}},b,\bot_{\mathit{cheat}})$ with weight $r$ for all $b \in B$. 
    
    Though the construction is more complicated than the one for the proof of \autoref{pro:warmup}, the transitions $(q_a,c,\bot_{\mathit{cheat}})$ still serve the same purpose. That is, they make it so that whenever Eve chooses state $q_a$ in the right sub-arena then Adam must play $a$. Additionally, when Eve chooses a state $\langle q'_a \rangle$, Adam is now able to play any observation $o \in \mathit{Obs}$. The transitions enforce alternation between these two cases so that Eve is able to force Adam's choice of letters but has no influence on his choice of observations --- even though she ``observes'' them as he spells the observations.
    
    \paragraph{($\Longrightarrow$)}
    Assume that all observation-based strategies of Player $1$ are such that there exists a run $\rho$ consistent with it and $\Val(\rho) > r$, i.e. Player $1$ has no $r$-winning observation-based strategy. In the initial gadget, if Eve goes to the left, then Adam plays $\mathit{bail}$ and her regret
    is strictly larger than $r$. If she goes right, the run proceeds into the right sub-arena. There, whenever Eve chooses state $q_a$, Adam plays $a$; whenever Eve is in a state $q'_a$ and Adam has already spelled $\alpha = a_0 o_1 \dots a_n$,
    Adam plays some $o \in \mathit{Obs}$ such
    that there is a (finite run) $\dots q_n$ on $\alpha$ in $\mathcal{G}$ with
    $\mathit{obs}(q_n) = o$. It follows that the run Eve constructs stays in the states $\{q_a,q'_a : a \in A \setminus \mathit{Obs}\}$ and has
    value $0$. Note that in the left sub-arena there is a run that does not reach $\bot$ and which is consistent with the strategy Eve has used. Indeed, since Eve only
    observes an observation-letter sequence (spelled by Adam) before choosing $q_b$ from $q'_a$, there is an observation-based strategy of Player $1$ in $\mathcal{G}$ which captures this
    behaviour. To conclude, recall we have assumed that all
    observation-based strategies of Player $1$ are such that there exists a run $\rho$ in $\mathcal{G}$
    consistent with it and $\Val(\rho) > r$. Hence, the regret of Eve is strictly larger than $r$.
    
    \paragraph{($\Longleftarrow$)}
    Assume now Player $1$ has an observation-based strategy $\sigma$ such that all runs $\rho$ consistent with it have $\Val(\rho) \leq r$. Indeed, we can assume $\sigma$ is such that all runs consistent with it avoid $\bot$. In the initial gadget, Eve will go to the right sub-arena. Hence, if Adam plays $\mathit{bail}$, the value of the run of Eve is optimal and her regret is $0$. Otherwise, if Adam does not play $\mathit{bail}$, the run constructed by Eve proceeds to the right sub-arena. There, after Adam has spelled $\alpha = a_0 o_1 \dots a_n o_n$, Eve moves to $q_b$ where $b = \sigma(\alpha)$. Note that the run she constructs stays in the states $\{q_a,q'_a : a \in A \setminus \mathit{Obs}\}$ if and only if Adam follows Eve's choices of letters. We consider both cases in turn below.
    \begin{description}
        \item[If Adam follows Eve's strategy] then Eve constructs a run on $\alpha$ with value $0$. By construction, all alternative runs on $\alpha$ in the right sub-arena have a value of at most $r$.
        Indeed, by construction $r$ is the maximal weight of any transition in this part of the arena.
        Recall we have also assumed that in $\mathcal{G}$ all runs $\rho$ consistent with $\sigma$ in the left sub-arena have a value smaller than or equal to $r$. If Adam chooses observations which make all runs reach $\bot$, the value of alternative runs therein is $0$. Otherwise, all runs that avoid $\bot$ in the left sub-arena thus have a value smaller than $r$. It follows that Eve has a regret value smaller than or equal to $r$.
        \item[If Adam does not follow Eve's strategy,] Eve's run reaches $\bot_{\mathit{cheat}}$ and thus has a value of $r$. Note that, by construction, this is the maximal value a run in the right sub-arena can have. In the left sub-arena, the maximal value a run can have is $r + 1$. To conclude we observe that:
        \[
            r + 1 - r = 1 < 3 < r,
        \]
        which means that Eve has a regret value smaller than $r$.\qedhere
    \end{description}
\end{proof}

\subsection{Memory requirements for \eve and \adam}
It is known that positional strategies suffice for \eve in parity games. On the
other hand, for Streett games she might require exponential memory (see,
e.g.~\cite{djw97}). This exponential blow-up, however, is only on the number of
pairs---which we have already argued remains polynomial w.r.t. the original
automaton. It follows that:
\begin{corollary}
	For payoff functions $\supfun$, $\inffun$, $\lsupfun$, $\linffun$,
	for all weighted automata
	$\mathcal{A}$, there exists $m$ in
	$2^{\mathcal{O}(|\mathcal{A}|)}$ such that:
	\[
		\Regret{\mathcal{A}}{\StrAllE,\StrWordA } =
		\Regret{\mathcal{A}}{\StrE^m,\StrWordA}.
	\]
\end{corollary}

\subsection{Fixed memory for \eve}
Since the problem is \EXP-hard for most payoff functions and already undecidable
for $\underline{\mpfun}$ and $\overline{\mpfun}$, we now fix the memory
\eve can use.

\begin{theorem}\label{thm:one-assumption}
	Given $r \in \mathbb{Q}$ and a weighted automaton $\Gamma$ with
	payoff function $\inffun$, $\supfun$, $\linffun$, $\lsupfun$,
	$\underline{\mpfun}$, or $\overline{\mpfun}$, determining whether
	\( \Regret{\Gamma}{\StrE^{m},\StrWordA} \lhd r \), for $\lhd \in
	\{<,\le\}$, can be done in $\NTIME(m^2|\Gamma|^2)$.
\end{theorem}
Denote by $\mathfrak{R}_\forall \subseteq \StrWordA$ the set of all word
strategies of \adam which are regular. That is to say, $w \in
\mathfrak{R}_\forall$ if and only if $w$ is \emph{ultimately periodic}.
It is well-known that the mean-payoff value of ultimately periodic plays in
weighted arenas is the same for both $\overline{\mpfun}$ and
$\underline{\mpfun}$.

Before proving the theorem we first show that ultimately periodic words suffice
for \adam to spoil a finite memory strategy of \eve. Let us fix some
useful notation. Given weighted automaton $\Gamma$ and a finite memory
strategy $\sigma$ for \eve in $\Gamma$ we denote by $\Gamma_\sigma$ the
deterministic automaton embodied by a refinement of $\Gamma$ that is induced by
$\sigma$.

\begin{lemma}\label{lem:reg-suffice}
	For $r \in \mathbb{Q}$, weighted automaton $\Gamma$, and
	payoff function $\inffun$, $\supfun$, $\linffun$, $\lsupfun$,
	$\underline{\mpfun}$, or $\overline{\mpfun}$, if \(
	\Regret{\Gamma}{\StrE^{m},\StrWordA} \rhd r \) then \(
	\Regret{\Gamma}{\StrE^{m}, \mathfrak{R}_\forall} \rhd r \), for
	$\rhd \in \{>,\ge\}$.
\end{lemma}
\begin{proof}
	For $\inffun$, $\supfun$, $\linffun$, and $\lsupfun$ the result follows
	from Lemma~\ref{lem:exp-memb}. It is known that positional strategies
	suffice for either player to win a parity game. Thus, if \adam wins the
	parity game defined in the proof of Lemma~\ref{lem:exp-memb} then he
	has a positional strategy to do so. Now, for any strategy of \eve in the
	original game, one can translate the winning strategy of \adam in the
	parity game into a spoiling strategy of \adam in the regret game. This
	strategy will have finite memory and will thus correspond to an
	ultimately periodic word. Hence, it suffices for us to show the claim
	follows for mean-payoff. We do so for $\underline{\mpfun}$ and $\ge$ but
	the result for $\overline{\mpfun}$ follows from minimal changes to the
	argument (a small quantifier swap in fact) and for $>$ variations we
	need only use the strict versions of Equations~\eqref{eqn:biweighted-1}
	and~\eqref{eqn:biweighted-2}.  We assume without loss of generality that all 
	weights are non-negative.

	Let $\sigma$ be the best (regret minimizing) strategy of \eve in
	$\Gamma$ which uses at most memory $m$. We claim that if \adam has a
	word strategy to ensure the regret of \eve in $\Gamma$ is at least $r$
	then he also has a regular word strategy to do so.

	Consider the bi-weighted graph $G$ constructed by taking the synchronous
	product of $\Gamma$ and $\Gamma_\sigma$ while labelling every edge with
	two weights: the value assigned to the transition by the weight function
	of $\Gamma_\sigma$ and the value assigned to the transition by that of
	$\Gamma$. For a path $\pi$ in $G$, denote by $w_i(\pi)$ the sum of the
	weights of the edges traversed by $\pi$ w.r.t. the $i$-th weight
	function. Also, for an infinite path $\pi$, denote by
	$\underline{\mpfun}_i$ the mean-payoff value of $\pi$ w.r.t.
	the $i$-th weight function. Clearly, \adam has a word strategy to ensure
	a regret of at least $r$ against the strategy $\sigma$ of \eve if and
	only if there is an infinite path $\pi$ in $G$ such that
	$\underline{\mpfun}_2(\pi) - \underline{\mpfun}_1(\pi) \ge r$. We
	claim that if this is the case then there is a simple cycle $\chi$ in
	$G$ such that $\frac{1}{|\chi|} w_2(\chi) - \frac{1}{|\chi|} w_1(\chi)
	\ge r$. The argument is based on the cycle decomposition of $\pi$ (see,
	e.g.~\cite{em79}). 
	
	Assume, for the sake of contradiction, that all the cycles $\chi$ in $G$
	satisfy the following: 
	\begin{equation}\label{eqn:biweighted-1}
		\frac{1}{|\chi|}w_2(\chi) -\frac{1}{|\chi|}w_1(\chi) \le r -
		\epsilon\text{, for some }0 < \epsilon \le r,\tag{$\ast$}
	\end{equation}
	and let us consider an arbitrary infinite path $\pi = v_0 v_1 \dots$.
	Let $l=\underline{\mpfun}_1(\pi)$.  We will show
	\begin{equation}\label{eqn:biweighted-2}
		\liminf_{k \to \infty}\frac{w_2(\langle v_j \rangle_{j \le k})}{k}-l \leq r-\epsilon,
	\end{equation}
	from which the required contradiction follows.

	For any $k\geq 0$, the
	cycle decomposition of $\langle v_j \rangle_{j \le k}$ tells us that apart from a small
	sub-path, $\pi'$, of length at most $n$ (the number of states in $G$), the
	prefix $\langle v_j \rangle_{j \le k}$ can be decomposed into simple cycles $\chi_1, \ldots, \chi_t$
	such that $w_i(\langle v_j \rangle_{j \le k}) = w_i(\pi') + \sum_{j=1}^t w_i(\chi_j)$ for $i=1,2$.  
	If $W$ is the maximum weight occurring in $G$, then
	from Equation~\eqref{eqn:biweighted-1} we have:
	\begin{eqnarray*}
		w_2(\langle v_j \rangle_{j \le k}) & \leq & nW + \sum_{j=1}^t w_2(\chi_j)\\
		& \leq & nW + (r-\epsilon)\sum_{j=1}^t |\chi_j| + \sum_{j=1}^t w_1(\chi_j)\\
		& \leq & nW + k(r-\epsilon) + w_1(\langle v_j \rangle_{j \le k}).
	\end{eqnarray*}
	Now, it follows from the definition of the limit inferior that for any
	$\epsilon'>0$ and any $K>0$ there exists $k>K$ such that $ w_1(\langle
	v_j \rangle_{j \le k}) \leq k(l+\epsilon')$.  Thus for any $\epsilon'>0$
	and $K'>0$, there exists $k>\max\{K',nW/\epsilon'\}$ such that
	\[
		\frac{w_2(\langle v_j \rangle_{j \le k})}{k} \leq \frac{nW}{k}
		+ (r - \epsilon) + (l + \epsilon') < (l+r-\epsilon) +
		2\epsilon'.
	\]
	Equation~\eqref{eqn:biweighted-2} then follows from the definition of limit inferior.
	
	The above implies that \adam can, by repeating $\chi$ infinitely often,
	achieve a regret value of at least $r$ against strategy $\sigma$ of
	\eve. As this can be done by him playing a regular word, the result
	follows.
\end{proof}

We now proceed with the proof of the theorem. The argument is presented for
mean-payoff ($\underline{\mpfun}$)  but minimal changes are required for the
other payoff functions. For simplicity, we use the non-strict threshold for the
emptiness problems. However, the result from~\cite{cdh10} is independent of this.
Further, the exact same argument presented here works for both cases. Thus, if
suffices to show the result follows for $\ge$.

\begin{proof}[Proof of Theorem~\ref{thm:one-assumption}]
We will ``guess'' a strategy for \eve which uses memory at most $m$ and verify
(in polynomial time w.r.t. $m$ and the size of $\Gamma$) that it
ensures a regret value of strictly less than $r$.

Let $\calA$ be the mean-payoff ($\underline{\mpfun}$) automaton constructed as the
synchronous product of $\Gamma$ and $\Gamma_\sigma$. The new weight function
maps a transition to the difference of the values of the weight functions of the
two original automata.  We claim that the language of $\calA$ is empty (for
accepting threshold $\ge r$) if and only if
$\regret{\sigma}{\Gamma}{\StrE^m,\StrWordA} < r$. Indeed, there is a bijective
map from every run of $\calA$ to a pair of plays $\pi,\pi'$ in $\Gamma$ such
that both $\pi$ and $\pi'$ are consistent with the same word strategy of \adam
and $\pi$ is consistent with $\sigma$.  It will be clear that $\calA$ has size
at most $m|\Gamma|$. As emptiness of a weighted automaton $\mathcal{A}$ can be
decided in $O(|\mathcal{A}|^2)$ time~\cite{cdh10}, the result will follow.

We now show that if the
language of $\calA$ is not empty then \adam can ensure a regret value of at
least $r$ against $\sigma$ in $\Gamma$ and that, conversely, if \adam has a
spoiling strategy against $\sigma$ in $\Gamma$ then that implies the language of 
$\calA$ is not empty.

Let $\rho_x$ be a run of $\calA$ on $x$. From the definition of $\calA$ we get
that $\underline{\mpfun}(\rho_x) = \liminf_{i \to \infty} \frac{1}{i}
\sum_{j=0}^{i} (a_j - b_j)$ where $\alpha_x = \langle a_i \rangle_{i \ge 0}$
and $\beta_x = \langle b_i \rangle_{i \ge 0}$
are the infinite sequences of weights assigned to the transitions of $\rho$ by
the weight functions of $\Gamma$ and $\Gamma_\sigma$ respectively. It is known
that if a mean-payoff automaton accepts a word $y$ then it must accept an
ultimately periodic word $y'$, thus we can assume that $x$ is ultimately
periodic (see, e.g.~\cite{cdh10}). Furthermore, we can also assume the run of
the automaton on $x$ is ultimately periodic. Recall that for ultimately periodic
runs we have that $\underline{\mpfun}(\rho_x) = \overline{\mpfun}(\rho_x)$. We
get the following
\begin{align*}
	\underline{\mpfun}(\rho_x) &= \limsup_{i \to \infty} \frac{1}{i}
	\sum_{j=0}^i (a_j - b_j)&\\
	&\le \limsup_{i \to \infty} \frac{1}{i}\sum_{j=0}^i a_j +
	\limsup_{i \to \infty} \frac{-1}{i}\sum_{j=0}^i
	b_j &\text{sub-additivity of }\limsup\\
	&\le \limsup_{i \to \infty} \frac{1}{i}\sum_{j=0}^i a_j -
	\liminf_{i \to \infty} \frac{1}{i}\sum_{j=0}^i
	b_j &\\
	&\le \liminf_{i \to \infty} \frac{1}{i}\sum_{j=0}^i a_j -
	\liminf_{i \to \infty} \frac{1}{i}\sum_{j=0}^i
	b_j &\text{ultimate periodicity}.
\end{align*}
Thus, as $x$ and $\rho_x$ can be be mapped to a strategy of \adam in $\Gamma$
which ensures regret of at least $r$ against $\sigma$, the claim follows.

For the other direction, assume \adam has a word strategy $\tau$ in $\Gamma$
which ensures a regret of at least $r$ against $\sigma$. From
Lemma~\ref{lem:reg-suffice} it follows that $\tau$ and the run $\rho$ of
$\Gamma$ with value $\Gamma(\tau)$ can be assumed to be ultimately periodic
w.l.o.g.. Denote by $\rho_\sigma$ and $w_\sigma$ the run of $\Gamma_\sigma$ on
$\tau$ and the weight function of $\Gamma_\sigma$ respectively. We then get that
\begin{align*}
	&\phantom{={}} \liminf_{i \to \infty} \frac{1}{i} w_\sigma(\rho_\sigma) -
	\liminf_{i \to \infty} \frac{1}{i} w(\rho)&\\
	&= \liminf_{i \to \infty} \frac{1}{i} w_\sigma(\rho_\sigma) +
	\limsup_{i \to \infty} \frac{-1}{i} w(\rho)&\\
	&= \liminf_{i \to \infty} \frac{1}{i} w_\sigma(\rho_\sigma) +
	\liminf_{i \to \infty} \frac{-1}{i} w(\rho)&\text{ultimate periodicity}\\
	&\le \underline{\mpfun}(\psi_\tau) &\text{super-additivity of }\liminf,
\end{align*}
where $\psi_\tau$ is the corresponding run of $\calA$ for $\tau$ and $\rho$.
Hence, $\calA$ has at least one word in its language.
\end{proof}

We provide a matching lower bound.  The proof is an adaptation of the
\NP-hardness proof from~\cite{akl10}.

\begin{theorem}\label{thm:np-hardness}
	Let $r \in \mathbb{Q}$ with $r > 0$, $\lhd \in
	\{<, \le\}$, and the payoff function be $\inffun$, $\supfun$, $\linffun$, $\lsupfun$,
	$\underline{\mpfun}$, or $\overline{\mpfun}$. Given a weighted automaton $\Gamma$, determining whether
	\( \Regret{\Gamma}{\StrE^{1},\StrWordA} \lhd r \),
	is \NP-hard.
\end{theorem}

\begin{figure}
\begin{center}
\begin{tikzpicture}[scale=0.7,initial text={},every state/.style={minimum size=0.5cm,inner sep=0}]
\node[state,initial above](I) at (2,5) {};
\node[state](A) at (0,4) {};
\node(B) at (2,4) {\dots};
\node[state](C) at (4,4) {};
\node[state](D) at (0,2) {};
\node(E) at (2,2) {\dots};
\node[state](F) at (4,2) {};
\node[state](G) at (2,1) {$\top$};

\path[->,auto]
(I) edge node[swap]{$1$} (A)
(I) edge node{$n$} (C)
(A) edge node[swap]{\#} (D)
(C) edge node{\#} (F)
(D) edge node[swap]{$1$} (G)
(F) edge node{$n$} (G)
(G) edge[loop below] node[swap]{$B,\frac{3r}{2}$} (G)
;

\end{tikzpicture}
\caption{Clause choosing gadget for the SAT reduction. There are as many paths
from top to bottom as there are clauses.}
\label{fig:det-clause-chooser}
\end{center}
\end{figure}

\begin{figure}
\begin{center}
\begin{tikzpicture}[yscale=0.7,initial text={},every state/.style={minimum size=0.5cm,inner sep=0}]
\node[state,initial above](I) at (3,6) {$q_0$};
\node[state](A) at (1,4) {$x_1$};
\node[state](At) at (0,2) {$\overline{x_1}$};
\node[state](Af) at (2,2) {$\underline{x_1}$};
\node[state](C) at (5,4) {$x_2$};
\node[state](Ct) at (4,2) {$\overline{x_2}$};
\node[state](Cf) at (6,2) {$\underline{x_2}$};
\node[state](G) at (3,-1) {$\bot_r$};

\path[auto,->]
(I) edge node[swap]{$1,2,3$} (A)
(I) edge node{$1,2,3$} (C)
(A) edge node[swap]{\#} (At)
(A) edge node{\#} (Af)
(C) edge node[swap]{\#} (Ct)
(C) edge node{\#} (Cf)
(At) edge node[swap]{$1$} (G)
(Af) edge node[pos=0.1,swap]{$2,3$} (G)
(Ct) edge node[pos=0.1]{$1,2$} (G)
(Cf) edge node{$3$} (G)
(G) edge[loop below] node[swap]{$B,r$} (G)
;

\end{tikzpicture}
\caption{Value choosing gadget for the SAT reduction. Depicted is the
configuration for $(x_1 \lor x_2) \land (\lnot x_1 \lor x_2)
\land(\lnot x_1 \lor \lnot x_2)$.}
\label{fig:nondet-value-chooser}
\end{center}
\end{figure}

\begin{proof}
	We give a reduction from the \textsc{SAT} problem, i.e. satisfiability
	of a CNF formula. The construction presented is based on a proof
	in~\cite{akl10}. The idea is simple: given a Boolean formula $\Phi$ in CNF
	we construct a weighted automaton $\Gamma_\Phi$ such that Eve can
	ensure a regret value of at most $r$ with a positional strategy in $\Gamma_\Phi$
	if and only if $\Phi$ is satisfiable. In symbols, we construct a weighted automaton $\Gamma_\Phi$ such that \( \Regret{\Gamma_\Phi}{\StrE^{1},\StrWordA} < r \) if $\Phi$ is satisfiable and \( \Regret{\Gamma_\Phi}{\StrE^{1},\StrWordA} > r \) otherwise. Notice the regret-value gap in the latter two cases. The result thus follows for both $\lhd \in \{<,\leq\}$. For concreteness, we consider the payoff function to be $\overline{\mpfun}$. At the end of the proof we comment on how to adapt the construction to prove the result for the other payoff functions.

	Let us now fix a Boolean formula $\Phi$ in CNF with $n$ clauses and $m$
	Boolean variables $x_1,\dots,x_m$. The weighted automaton
	$\Gamma_\Phi = (Q, q_I, B, \Delta, w)$ has alphabet $B = \{bail,\#\}
	\cup \{1 \le i \le n\}$. It includes as initial gadget
	precisely the one depicted in Figure~\ref{fig:initial-gadget-pompg}. Recall that
	this gadget forces Eve to play into the right sub-arena --- otherwise, she has regret of at least $r+1$. As the left
	sub-arena of $\Gamma_\Phi$ we attach the gadget depicted in
	Figure~\ref{fig:det-clause-chooser}. All transitions shown have weight
	$0$ and all missing transitions, in order for $\Gamma_\Phi$ to be
	complete, lead to the state $\bot_0$ from the initial gadget. Intuitively, as Eve must go to the right sub-arena
	then all alternative plays in the left sub-arena correspond to either
	Adam choosing a clause $i$ and spelling $i \# i$ to reach $\top$ or
	reaching $\bot_0$ by playing any other sequence of symbols. The right
	sub-arena of the automaton is as shown in
	Figure~\ref{fig:nondet-value-chooser}, where all transitions shown have
	weight $0$ and all missing transitions go to $\bot_{\nicefrac{r}{3}}$, a sink state with $\nicefrac{r}{3}$-weighted self-loops on all $b \in B$. Here, from
	$q_0$ we have transitions to state $x_j$ with symbol $i$ if the $i$-th
	clause contains the variable $x_j$. For every state $x_j$ we have
	transitions to $\overline{x_j}$ and $\underline{x_j}$ with symbol $\#$. The idea is
	to allow Eve to choose the truth value of $x_j$. Finally, every state
	$\overline{x_j}$ (and $\underline{x_j}$) has a transition to $\bot_r$ with symbol $i$ if
	the literal corresponding to the state appears in the $i$-th
	clause. The state $\bot_r$ has ${r}$-weighted self-loops on all $b \in B$.

    Assume the formula is indeed
	satisfiable. We will argue that in this case Eve has a regret value that is strictly less than $r$.
	\begin{description}
    \item[If Adam chooses $1 \le i \le n$
	and spells $i \# i$,] since we know $\Phi$ is satisfiable, there is a choice of
	values for $x_1,\ldots,x_m$ such that for each clause
	there must be at least one literal $\ell$ in the $i$-th clause which makes the
	clause true. Eve transitions into the right sub-arena. Then, from $q_0$, she transitions to the state corresponding to $\ell$
	and when Adam plays $\#$ she chooses the correct
	truth value for the variable. Thus, the play reaches $\bot_r$ and has a value of $r$. Since all alternative plays have a value of at most $\nicefrac{3r}{2}$, we conclude that Eve has a regret of at most $\nicefrac{3r}{2} - r = \nicefrac{r}{2} < r$.
	\item[If Adam does not
	play $i \# i$, for some $i$,] then all plays will have value at most $r$ and at least $\nicefrac{r}{3}$. Thus, the regret of Eve will be at most $\nicefrac{2r}{3} < r$.
	\end{description}
	Note that the strategy for Eve can be realized
	with a positional strategy by assigning to each $x_j$ the choice of
	truth value and choosing from $q_0$ any valid transition for all $1 \le
	i \le n$.
	
	Conversely, if $\Phi$ is not satisfiable, for all valuations of the
	variables $x_1,\dots, x_m$ there is at least one clause which is not
	true. Given any positional strategy of Eve in $\Gamma_\Phi$, we can
	extract the corresponding valuation of the Boolean variables. Now, Adam
	chooses $1 \le i\le n$ such that the $i$-th clause is not satisfied by
	the assignment. The play will therefore end in $\bot_{\nicefrac{r}{3}}$ while an
	alternative play in the left sub-arena will reach $\top$. Hence the
	regret of Eve in the game is at least $\nicefrac{3r}{2} - \nicefrac{r}{3} = \nicefrac{7r}{6} > r$.

	To complete the proof, we note that the above analysis works the same for
	all payoff functions except for $\inffun$. However, changing the weight of all $0$-weight transitions (except the one in the initial gadget) to $r$, allows us to repeat the exact same argument for $\inffun$.
\end{proof}

\subsection{Relation to other works}
Let us first extend the definitions of \emph{approximation, embodiment} and
\emph{refinement} from~\cite{akl10} to the setting of $\omega$-words.  Consider
two weighted automata $\mathcal{A} = (Q_\calA,q_I,A,\Delta_\calA,w_\calA)$ and
$\mathcal{B} = (Q_\calB, q_I, A, \Delta_\calB,w_\calB)$ and let $d : \mathbb{R}
\times \mathbb{R} \to \mathbb{R}$ be a \emph{metric}.\footnote{The
metric used in~\cite{akl10} is the ratio measure.} 
We say \emph{$\mathcal{B}$ (strictly) $\alpha$-approximates $\mathcal{A}$ (with
respect to $d$)} if $d(\mathcal{B}(w), \mathcal{A}(w)) \leq \alpha$ (resp.
$d(\mathcal{B}(w), \mathcal{A}(w)) < \alpha$) for all words $w \in
A^\omega$. We say \emph{$\mathcal{B}$ embodies $\mathcal{A}$} if $Q_\calA
\subseteq Q_\calB$, $\Delta_\calA \subseteq \Delta_\calB$ and $w_\calA$ agrees
with $w_\calB$ on $\Delta_\calA$. For an automaton $\mathcal{A} =
(Q,q_I,A,\Delta,w)$ and an integer $k \ge 0$, the $k$-refinement of
$\mathcal{A}$ is the automaton obtained by refining the state space of
$\mathcal{A}$ using $k$ Boolean variables. Intuitively, this corresponds to
having $2^k$ copies of every state, with each copy of $p$ transitioning to all
copies of $q$ with $a$ if $(p,a,q) \in \Delta$. The automaton $\mathcal{A}$ is said
to be \emph{(strictly) $(\alpha,k)$-determinizable by pruning} (DBP, for short) if the
$k$-refinement of $\mathcal{A}$ embodies a deterministic automaton which
(strictly) $\alpha$-approximates $\mathcal{A}$. The next result follows directly
from the above definitions.

\begin{proposition}\label{pro:rel-dbp}
	For $\alpha \in \mathbb{Q}$, $k \in \mathbb{N}$, a weighted
	automaton $\Gamma$ is (strictly) $(\alpha,k)$-DBP (w.r.t. the difference
	metric) if and only if \( \Regret{\Gamma}{\StrE^{2^k},\StrWordA} \leq \alpha \)
	(resp. $\Regret{\Gamma}{\StrE^{2^k},\StrWordA} < \alpha$).
\end{proposition}

In~\cite{hp06} the authors define \emph{good for games automata}. Their
definition is based on a game which is played on an $\omega$-automaton by
Spoiler and Simulator. We propose the following generalization of the notion of
good for games automata for weighted automata. A
weighted automaton $\mathcal{A}$ is \emph{(strictly) $\alpha$-good for games} if
Simulator, against any word $w \in A^\omega$ spelled by Spoiler, can resolve
non-determinism in $\mathcal{A}$ so that the resulting run has value $v$ and
$d(v, \mathcal{A}(w)) \leq \alpha$ (resp. $d(v, \mathcal{A}(w)) < \alpha)$, for
some metric $d$.  We summarize the relationship that follows from the definition
in the following result:
\begin{proposition}\label{pro:rel-gfg}
	For $\alpha \in \mathbb{Q}$, a weighted automaton $\Gamma$
	is (strictly) $\alpha$-good for games (w.r.t. the difference metric)
	if and only if $\Regret{\Gamma}{\StrAllE,\StrWordA} \leq \alpha$ (resp.
	$\Regret{\Gamma}{\StrAllE,\StrWordA} < \alpha$).
\end{proposition}

\section{Discussion}
In this work we have considered the regret threshold problem in quantitative
games. We have studied three variants which corresponds to different assumptions
regarding the behavior of \adam. Our definition of regret is based on the
difference measure: \eve attempts to minimize the difference between the value
she obtains by playing the game, and the value she could have obtained if she
had known the strategy of \adam in advance. In~\cite{akl10} the ratio measure
was used instead. We believe some of the results obtained presently can be
extended to arbitrary metrics (as in, \eg,~\cite{bh14}). In particular, all
hardness statements should hold. We give more precise claims for the ratio
measure below.

\paragraph{For $\inffun$, $\supfun$, $\linffun$, and $\lsupfun$.}~
We have already observed that upper bounds for the regret threshold problem
follow directly from our results if regret is defined using ratio (see
Remark~\ref{rem:metric}). Furthermore, all hardness results presented here can
also be adapted to obtain the same result for ratio. Indeed, the same
constructions and gadgets can be used. These, together with correctly chosen
regret threshold value $r$ and modified edge weights and inequalities
(such as the ones given to prove, for instance,
Lemma~\ref{lem:pspace-hardness}) are sufficient to show the same results hold
for regret defined with ratio.

\paragraph{For $\mpfun$.}
All hardness results also hold for regret defined with ratio. As with the other
payoff functions, minimal modifications are needed for the proofs given in this
work to imply the result for the alternative definition of regret. Regarding the
algorithms, we have solved the regret threshold problem for the first two
variants. In the third variant, we have considered a restricted version of the
game (Theorem~\ref{thm:one-assumption}) and given an algorithm for it by
reducing it to an the emptiness problem for mean-payoff automata. We claim the
corresponding problems are in the same complexity classes, respectively, when
regret is defined with ratio. For the first two, the proofs are almost identical
to the ones we have give in the present work for the difference measure. For the
third problem, Lemma~\ref{lem:reg-suffice} must be restated for ratio, yet the
proof requires minimal modifications to work in that case. Finally, the argument
used to prove Theorem~\ref{thm:one-assumption} requires the reduction to
mean-payoff automata be replaced by a reduction to ratio automata. However, all
the properties from mean-payoff automata which were used in the proof, are also
true for ratio automata (\eg~ultimately periodic words being accepted if an
arbitrary word is accepted). The latter follow from results regarding ratio
games in~\cite{bcghhjkk14}.

\section*{Acknowledgements}
We thank Udi Boker for his comments on how to determinize $\lsupfun$ automata. We further thank him and Karoliina Lehtinen for notifying us about inaccuracies in some of the results stated for Variant III (playing against word strategies).

\bibliographystyle{alpha}
\bibliography{refs}
\end{document}